%% file: main.tex
\documentclass[11pt]{article}
\usepackage{fullpage}
\usepackage{graphicx}
\usepackage{amsmath,amssymb,amsthm,amsfonts,dsfont,mathtools,latexsym}
\usepackage{relsize}
\usepackage[margin=1cm]{caption}
\usepackage{wrapfig}
\usepackage{frame,color}
\usepackage{environ,subfig}
\usepackage{hyperref}
\usepackage{xspace}
\usepackage{framed}
\usepackage{comment}
\usepackage{indentfirst}
\usepackage{enumerate}
\theoremstyle{plain}
\newtheorem{theorem}{Theorem}[section]
\newtheorem{lemma}[theorem]{Lemma}

\newtheorem{corollary}[theorem]{Corollary}

\theoremstyle{definition}
\newtheorem{definition}{Definition}[section]

\newcommand{\ignore}[1]{}
\newcommand{\E}{\operatorname{E}}
\newcommand{\Var}{\operatorname{Var}}

\newcommand{\argmin}{\operatorname{argmin}}

\newcommand{\sqbrack}[1]{\left[ #1 \right]}

\newcommand{\poly}{{\operatorname{poly}}}
\newcommand{\nnz}{{\operatorname{nnz}}}
\newcommand{\rank}{{\operatorname{rank}}}
\newcommand{\polylog}{{\operatorname{polylog}}}
\newcommand{\hardA}{\mathcal{A}}
\newcommand{\normG}{\mathcal{N}(0, 1)}

\title{Tight Bounds for $\ell_p$ Oblivious Subspace Embeddings}
\author{Ruosong Wang\\ Carnegie Mellon University \\ \footnotesize \texttt{ruosongw@andrew.cmu.edu}\and
David P. Woodruff\\ Carnegie Mellon University \\ \footnotesize \texttt{dwoodruf@cs.cmu.edu}
}
\date{}
\begin{document}
\begin{titlepage}
\maketitle
\thispagestyle{empty}
\begin{abstract}
An $\ell_p$ oblivious subspace embedding is a distribution over
$r \times n$ matrices $\Pi$ such that for any fixed $n \times d$ matrix $A$,
$$\Pr_{\Pi}[\textrm{for all }x, \ \|Ax\|_p \leq \|\Pi Ax\|_p \leq \kappa \|Ax\|_p] \geq 9/10,$$
where $r$ is the {\it dimension} of the embedding, $\kappa$ is the {\it distortion}
of the embedding, and for an $n$-dimensional vector $y$, 
$\|y\|_p = \left (\sum_{i=1}^n |y_i| \right )^{1/p}$ is the $\ell_p$-norm. Another important property
is the {\it sparsity} of $\Pi$, that is, the maximum number of non-zero entries per column,
as this determines the running time of computing $\Pi \cdot A$. While for $p = 2$ there are nearly
optimal tradeoffs in terms of the dimension, distortion, and sparsity, for the important case of $1 \leq p < 2$,
much less was known. In this paper we obtain nearly optimal tradeoffs for $\ell_p$ oblivious subspace embeddings for every $1 \leq p < 2$. Our main results
are as follows:
\begin{enumerate}
\item We show for every $1 \leq p < 2$, any oblivious subspace embedding with dimension $r$ has
distortion $\kappa = \Omega \left(\frac{1}{\left(\frac{1}{d}\right)^{1 / p} \cdot \log^{2 / p}r + \left(\frac{r}{n}\right)^{1 / p - 1 / 2}}\right).$ When $r = \poly(d) \ll n$ in applications, this gives 
a $\kappa = \Omega(d^{1/p}\log^{-2/p} d)$ lower bound, and shows the oblivious subspace
embedding of Sohler and Woodruff (STOC, 2011) for $p = 1$ and the oblivious subspace embedding of
Meng and Mahoney (STOC, 2013) for $1 < p < 2$ are optimal up to $\poly(\log(d))$ factors. 
\item We give sparse oblivious subspace embeddings for every $1 \leq p < 2$ which are
optimal in dimension and distortion, up to $\poly(\log d)$ factors. Importantly for $p = 1$, 
we achieve $r = O(d \log d)$, $\kappa = O(d \log d)$ and $s = O(\log d)$ non-zero entries per column.
The best previous construction with $s \leq \poly(\log d)$ is due to Woodruff and Zhang (COLT, 2013),
giving $\kappa = \Omega(d^2 \poly(\log d))$ or $\kappa = \Omega(d^{3/2} \sqrt{\log n} \cdot \poly(\log d))$ and $r \geq d \cdot \poly(\log d)$; 
in contrast our $r = O(d \log d)$ and $\kappa = O(d \log d)$ are optimal
up to $\poly(\log(d))$ factors even for dense matrices. 
\end{enumerate}
We also give (1) nearly-optimal  $\ell_p$ oblivious subspace embeddings with an
expected $1+\varepsilon$ number of non-zero entries per column for arbitrarily small $\varepsilon > 0$, 
and (2) the first oblivious subspace embeddings for $1 \le p < 2$ 
with $O(1)$-distortion and dimension independent of $n$. Oblivious subspace embeddings are crucial
for distributed and streaming environments, as well as entrywise $\ell_p$ low rank
approximation. Our results give improved algorithms for these applications.  
\end{abstract}
\end{titlepage}
\input{intro}
\input{preliminary}
\input{hardness}
\input{discussion_hardness}

\input{l1}
\input{lp}
\input{sparsity}
\bibliography{bib}
\input{appendix}
\end{document}

%% file: intro.tex
\section{Introduction}
An $\ell_p$ {\em oblivious subspace embedding} with distortion $\kappa$ is a distribution over $r \times n$ matrices $\Pi$ such that for any given $A \in \mathbb{R}^{n \times d}$,
with constant probability, $\|Ax\|_p \le \|\Pi A x\|_p \le \kappa \|Ax\|_p$ simultaneously for all $x \in \mathbb{R}^d$.
The goal is to minimize $r$, $\kappa$ and the time to calculate $\Pi A$.

Oblivious subspace embeddings have proven to be an essential ingredient for approximately solving numerical linear algebra problems, such as regression and low-rank approximation.
S\'arlos \cite{sarlos2006improved} first used $\ell_2$ oblivious subspace embeddings to solve $\ell_2$-regression and Frobenius-norm low-rank approximation.
To see the connection, suppose one wishes to solve the $\ell_2$-regression problem $\argmin_x \|Ax-b\|_2$ in the overconstrained setting, i.e., $A \in \mathbb{R}^{n \times d}$ and $b \in \mathbb{R}^n$ where $n \gg d$. 
S\'arlos showed that in order to solve this problem approximately, it suffices to solve a much smaller instance $\argmin_x \|\Pi Ax- \Pi b\|_2$, provided $\Pi$ is an $\ell_2$ oblivious subspace embedding. 
S\'arlos further showed that using the Fast Johnson-Lindenstrauss Transform in \cite{ailon2006approximate} as the $\ell_2$ oblivious subspace embedding with $\kappa = 1 + \varepsilon$, one can get a $(1 + \varepsilon)$-approximate solution to the $\ell_2$-regression problem in $O(n d \log d) + \poly(d / \varepsilon)$ time, which is a substantial improvement over the standard SVD-based approach which runs in $O(nd^2)$ time.

Subsequent to the work of S\'arlos, the ``sketch and solve'' approach became an important way to solve numerical linear algebra problems.
We refer interested readers to the monograph of Woodruff \cite{woodruff2014sketching} for recent developments. 

The bottleneck of S\'arlos's approach is the step to calculate $\Pi A$, which requires $\Omega(nd \log d)$ time due to the structure of the Fast Johnson-Lindenstrauss Transform.
Although this is already nearly-optimal for dense matrices, when $A$ is large and sparse, one may wish to solve the problem faster than $O(nd)$ time by exploiting the sparsity of $A$.
Clarkson and Woodruff \cite{clarkson2013low} showed that there exist $\ell_2$ oblivious subspace embeddings with $r = \poly(d / \varepsilon)$ rows, $s = 1$ non-zero entries per column, and $\kappa = 1 + \varepsilon$.
The property that $s = 1$ is significant, since it implies calculating $\Pi A$ requires only $O(\nnz(A))$ time, where $\nnz(A)$ is the number of non-zero entries of $A$. 
In fact, the oblivious subspace embedding they used is the \textsf{CountSketch} matrix from the data stream literature \cite{charikar2002finding}.
By using the the \textsf{CountSketch} embedding in \cite{clarkson2013low}, one can reduce an $\ell_2$-regression instance of size $n \times d$ into a smaller instance of size $\poly(d / \varepsilon) \times d$ in $O(\nnz(A))$ time.
The original proof in \cite{clarkson2013low} used a technique based on splitting coordinates by leverage scores. 
The number of rows can be further reduced to $r = O((d / \varepsilon)^2)$ using the same construction and a finer analysis based on second moment method, shown independently in \cite{meng2013low} and \cite{nelson2013osnap}.

One may wonder if it is possible to further reduce the number of rows in the \textsf{CountSketch} embedding, since this affects the size of the smaller instance to solve.
In \cite{nelson2013sparsity}, Nelson and Nguy$\tilde{\hat{\mbox{e}}}$n showed that any $\ell_2$ oblivious subspace embedding with constant distortion and $s = 1$ non-zero entries per column requires $\Omega(d^2)$ rows.
Although this rules out the possibility of further reducing the number of rows in the \textsf{CountSketch} embedding, this lower bound can be circumvented by considering embeddings with $s > 1$ non-zero entries in each column.
This idea is implemented by the same authors in \cite{nelson2013osnap}, obtaining a result showing that for any $B > 2$, for $r$ about $B\cdot d \log^8 d/\varepsilon^2$ and $s$ about $\log^3_B d/\varepsilon$, one can achieve an $\ell_2$ oblivious subspace embedding with $\kappa = 1 + \varepsilon$.
The bound on $r$ and $s$ was further improved in \cite{cohen2016nearly} (see also \cite{bourgain2015toward}), where Cohen showed that for any $B > 2$, it suffices to have $r = O(B \cdot d \log d / \varepsilon^2)$ and $s = O(\log_Bd / \varepsilon)$.
Cohen's result matches the lower bound in \cite{nelson2014lower} up to a multiplicative $\log d$ factor in the number of rows. 

Another line of research focused on the case when $p \neq 2$, as the corresponding regression and low rank approximation problems
are often considered to be more robust, or less sensitive to outliers. 
Moreover, the $p = 1$ error measure for regression yields the maximum likelihood estimator under Laplacian noise models. When $p = 1$, using Cauchy random variables, Sohler and Woodruff \cite{sohler2011subspace} showed there exist $\ell_1$ oblivious subspace embeddings with $O(d \log d)$ rows and $\kappa = O(d \log d)$.
This approach was generalized by using $p$-stable random variables in work of Meng and Mahoney \cite{meng2013low} to $\ell_p$-norms when $1 < p < 2$, where they showed there exist $\ell_p$ oblivious subspace embeddings with $O(d \log d)$ rows and $\kappa = O\left((d \log d)^{1 / p}\right)$.
Unlike the case when $p = 2$, due to the large distortion incurred in such upper bounds, one cannot directly get a $(1 + \varepsilon)$-approximate solution to the $\ell_p$-regression problem by solving $\argmin_{x} \|\Pi A x - \Pi b\|_p$. A natural question then,
is whether it is possible to obtain $(1+\varepsilon)$-distortion with $\ell_p$ oblivious subspace embeddings; prior to our work
there were no lower bounds ruling out the simplest of algorithms for $p \neq 2$: 
(1) compute $\Pi A$ and $\Pi b$, and (2) output $\argmin_{x} \|\Pi A x - \Pi b\|_p$. 

Although it was unknown if better oblivious subspace embeddings exist for $p \neq 2$ prior to our work, 
$\ell_p$ oblivious subspace embeddings still played a crucial role in solving $\ell_p$-regression problems in earlier work, 
since they provide a way to 
{\em precondition} the matrix $A$, which enables one to further apply non-oblivious (sampling-based) subspace embeddings.
We refer interested readers to Chapter 3 of \cite{woodruff2014sketching} and references therein for further details.
Recent developments in entrywise $\ell_p$ low-rank approximation \cite{song2017low} also used $\ell_p$ oblivious subspace embeddings as an important ingredient. 
Furthermore, such $\ell_1$ oblivious subspace embeddings are the only known way to achieve
single-pass streaming algorithms for $\ell_1$-regression (see, e.g., Section 5 of \cite{sohler2011subspace}, where it is shown how to implement
the preconditioning and sampling in parallel in a single pass), a model that has received
considerable interest for linear algebra problems (see, e.g., \cite{cw09}). We note that recent algorithms for $\ell_p$-regression
based on Lewis weight sampling require at least $\Omega(\log \log n)$ passes in the streaming model. 

Due to these applications, 
speeding up the computation of $\Pi A$ for $\ell_p$ oblivious subspace embeddings is an important goal. 
In \cite{clarkson2016fast}, Clarkson et al. combined the idea of Cauchy random variables and Fast Johnson-Lindenstrauss Transforms
 to obtain a more structured family of subspace embeddings, which enables one to calculate $\Pi A$ in $O(nd \log n)$ time.
Meng and Mahoney \cite{meng2013low} showed that when $1 \le p < 2$, there exist $\ell_p$ oblivious subspace embeddings 
with $r = \widetilde{O}(d^5)$ rows and $s = 1$ non-zero entries per column, where the distortion $\kappa = \widetilde{O}(d^{3 / p})$.
The structure of the embedding by Meng and Mahoney is very similar to the \textsf{CountSketch} embedding by Clarkson and Woodruff 
\cite{clarkson2013low}.
In fact, to prove the distortion bound, Meng and Mahoney also used techniques of splitting coordinates based on leverage scores. 

Inspired by the technique by Andoni in \cite{andoni2017high}, which used exponential random variables to estimate the $\ell_p$-norm of a data steam, Woodruff and Zhang \cite{woodruff2013subspace} improved the embedding given in \cite{meng2013low}.
They showed there exist $\ell_1$ oblivious subspace embeddings with $r = \widetilde{O}(d)$ rows and $s = \polylog(d)$ non-zero entries per column, where the distortion $\kappa = \min\{\widetilde{O}(d^2), \widetilde{O}(d^{1.5}) \sqrt{\log n}\}$. Note that to achieve such a small polylogarithmic
sparsity, the distortion $\kappa$ had to either increase to $\widetilde{O}(d^2)$ or to $\widetilde{O}(d^{1.5}) \sqrt{\log n}$, the latter also
depending on $n$.  

The above works leave many gaps in our understanding on the tradeoffs between dimension, distortion, and sparsity for $\ell_p$ oblivious
subspace embeddings. For instance, 
it is natural to ask what the optimal distortion bound for $\ell_p$ oblivious subspace embeddings is when $1 \le p < 2$, 
provided that the number of rows $r = \poly(d)$.
Results in \cite{sohler2011subspace, meng2013low} showed that $\kappa = O\left((d \log d)^{1 / p}\right)$ is achievable. Is this optimal?
Also, it is unknown whether there exist {\em sparse} $\ell_p$ oblivious subspace embeddings with dimension $\widetilde{O}(d)$ and 
distortion $\kappa = \widetilde{O}\left(d^{1 / p}\right)$.
In this paper, we resolve these questions. 
\subsection{Our Results}
\paragraph{Distortion Lower Bound.} 
We first show a distortion lower bound for $\ell_p$ oblivious subspace embeddings, when $1 \le p < 2$.
\begin{theorem}\label{thm:main_lb}
For $1 \le p < 2$, if a distribution over $r \times n$ matrices $\Pi$ is an $\ell_p$ oblivious subspace embedding, then the distortion
$$
\kappa = \Omega \left(\frac{1}{\left(\frac{1}{d}\right)^{1 / p} \cdot \log^{2/p}r + \left(\frac{r}{n}\right)^{1 / p - 1 / 2}}\right).
$$
\end{theorem}
When $1 \le p < 2$ and $r = \poly(d)$, the denominator of the lower bound is dominated by the $\left(\frac{1}{d}\right)^{1 / p} \cdot \log^{2 / p}r$ term, provided $n$ is large enough.
In that case, our lower bound is $\Omega(d^{1 / p}  \log^{-2/p} d)$.
It is shown in \cite{sohler2011subspace} (for $p = 1$) and Theorem 6 in \cite{meng2013low} (for $1 < p < 2$) that there exist 
$\ell_p$ oblivious subspace embeddings with $r = O(d \log d)$ rows and distortion $\kappa = O\left((d \log d)^{1 / p}\right)$.
Our lower bound matches these results up to an $O\left(\log^{3/p}d\right)$ factor.
Thus, our lower bound is nearly-optimal for $r = \poly(d)$ (which is the main regime of interest in the above applications). 

The dependence on $(r / n)^{1 / p - 1 / 2}$ reflects the fact that
\begin{itemize}
\item When the number of rows $r = n$, one can get a trivial $\ell_p$ oblivious subspace embedding with $\kappa = 1$, i.e., the identity matrix $I$;
\item As $p \to 2$, there exist $\ell_2$ oblivious subspace embeddings \cite{sarlos2006improved, clarkson2013low, meng2013low, nelson2013osnap, bourgain2015toward, cohen2016nearly} with $\kappa = 1 + \varepsilon$ and $r = \poly(d / \varepsilon)$, where $\varepsilon$ can be an arbitrarily small constant.
\end{itemize}

It is possible that the $\log^{2 / p}r$ factor (in the $(1 / d)^{1 / p} \cdot \log^{2 / p}r$ term) could be somewhat improved. 
However, we show that some dependence on $r$ is in fact necessary. 

\begin{theorem}[Informal version of Theorem \ref{thm:dependence_r}]\label{thm:dependence_r_intro}
For $1 \le p < 2$, there exists an $\ell_p$ oblivious subspace embedding
over $\exp(\exp(O(d))) \times n$ matrices $\Pi$, 
where the distortion $\kappa$ is a constant.
\end{theorem}
Even though Theorem \ref{thm:dependence_r_intro} has a doubly exponential dependence on $d$ in the number of rows, 
it is the first $\ell_p$ oblivious subspace embedding with constant distortion,
when $1 \le p < 2$ and $r$ does not depend on $n$.
This new embedding suggests that it is impossible to get a lower bound of
$$
\kappa = \Omega \left(\frac{1}{\left(\frac{1}{d}\right)^{1 / p} + \left(\frac{r}{n}\right)^{1 / p - 1 / 2}}\right),
$$
i.e., the $\left(\frac{1}{d}\right)^{1 / p}$ term should have some dependence on $r$.
\paragraph{New $\ell_p$ oblivious subspace embeddings.}
We next show there exist {\em sparse} $\ell_p$ oblivious subspace embeddings with nearly-optimal distortion.
\begin{theorem}[Summary of Theorem \ref{thm:main_l1_cs}, \ref{thm:main_l1_osnap}, \ref{thm:main_lp_cs} and \ref{thm:main_lp_osnap}.]\label{thm:main_lp_intro}
For $1 \le p < 2$, there exist $\ell_p$ oblivious subspace embeddings over $r \times n$ matrices $\Pi$
with $s$ non-zero entries per column and distortion $\kappa$, where
\begin{enumerate}
\item When $p = 1$, 
\begin{enumerate}
\item $r = O(d^2)$, $s = 2$ and $\kappa = O(d)$; or
\item For sufficiently large $B$, $r = O(B \cdot d \log d)$, $s = O(\log_B d)$ and $\kappa = O(d \log_B d)$.
\end{enumerate}
\item When $1 < p < 2$, $\kappa = O\left((d \log d)^{1 / p}\right)$,
\begin{enumerate}
\item $r = O(d^2)$, $s = 2$; or
\item For sufficiently large $B$, $r = O(B \cdot d \log d)$, $s = O(\log_B d)$.
\end{enumerate}
\end{enumerate}
\end{theorem}
Notably, the distortion of our embeddings is never worse than the dense constructions in \cite{sohler2011subspace, meng2013low}.

Also, when $p = 1$, if we set $r = O(d^2)$ (Case 1(a)) or $r = O(d^{1 + \eta})$ for any constant $\eta > 0$ (Case 1(b)), 
then the distortion can be further improved to $O(d)$.
This is the first known $\ell_1$ oblivious subspace embedding with $r = \poly(d)$ rows and distortion $\kappa = o(d \log d)$. 
We remark that by using the dense construction in \cite{sohler2011subspace}, it is also possible to reduce the distortion to $O(d)$ by increasing the number of rows.

Similar to the \textsf{OSNAP} embedding in \cite{nelson2013osnap}, our results in Case 1(b) and Case 2(b) provide a tradeoff 
between the number of rows and the number of non-zero entries in each column.

\paragraph{Sparser $\ell_p$ oblivious subspace embeddings.}
Finally, we show that the sparsity of Case 1(a) and Case 2(a) in Theorem \ref{thm:main_lp_intro} 
can be further reduced by using two different approaches. 

The first approach is based on random sampling, which leads to the following theorem.
\begin{theorem}[Summary of Theorem \ref{thm:main_l1_sparse}{} and \ref{thm:main_lp_sparse}]\label{thm:main_sparse_intro}
For $1 \le p < 2$ and any constant $0 < \varepsilon < 1$, there exists an $\ell_p$ oblivious subspace embedding over $O(d^2) \times n$ matrices $\Pi$ 
where each column of $\Pi$ has at most two non-zero entries and $1 + \varepsilon$ non-zero entries in expectation, and the 
distortion $\kappa = O(d)$ (when $p = 1$)
or $\kappa = O\left((d \log d)^{1 / p}\right)$ (when $1 < p < 2$).
\end{theorem}
The second approach is based on the construction in \cite{meng2013low} and a truncation argument, which leads to the following theorem.
\begin{theorem}[Summary of Theorem \ref{thm:main_l1_trunc} and \ref{thm:main_lp_trunc}]\label{thm:main_trunc_intro}
For $1 \le p < 2$, there exists an $\ell_p$ oblivious subspace embedding over $\widetilde{O}(d^4) \times n$ matrices $\Pi$ 
where each column of $\Pi$ has a single non-zero entry and distortion $\kappa = \widetilde{O}(d^{1 / p})$.
\end{theorem}
It has been shown in \cite{nelson2013sparsity} that for any distribution over $r \times n$ matrices $\Pi$ with $s = 1$ non-zero entries 
per column, if for any fixed matrix $A \in \mathbb{R}^{n \times d}$, $\rank(\Pi A) = \rank(A)$ with constant probability, 
then $\Pi$ should have $r = \Omega(d^2)$ rows.
Since oblivious subspace embeddings with finite distortion always preserve the rank, this lower bound can also be applied. We show also
that this lower bound holds even if the columns of $\Pi$ have $1+\varepsilon$ non-zero entries in expectation for a constant $0 < \varepsilon < 1$,
thereby showing Theorem \ref{thm:main_sparse_intro} is optimal. 

We leave it as an open question to obtain $\ell_p$ oblivious subspace embedding with $r = O(d^2)$ rows, $s = 1$ non-zero entries in every column (as opposed to $1+\varepsilon$ in expectation), and $\kappa = \widetilde{O}\left(d^{1 / p}\right)$, or prove a stronger lower bound. 
If one insists on having $s = 1$ non-zero entries per column, then Theorem \ref{thm:main_trunc_intro} can be applied, at the cost of increasing the number of rows to $\widetilde{O}(d^4)$.
\input{comparison}
\input{application}
\subsection{Our Techniques}
\paragraph{Distortion lower bound.} 
We use the case when $p = 1$ to illustrate our main idea for proving our distortion lower bounds.
We start with Yao's minimax principle which enables us to deal only with deterministic embeddings. 
Here our goal is to construct a distribution over matrices $A \in \mathbb{R}^{n \times d}$ such that for any $\Pi \in \mathbb{R}^{r \times n}$, if
\begin{equation} \label{equ:intro}
\frac{1}{\kappa}\|Ax\|_1 \le \|\Pi A x\|_1 \le \|Ax\|_1
\end{equation}
holds simultaneously for all $x \in \mathbb{R}^d$ with constant probability, then $\kappa = \widetilde{\Omega}(d)$.

Roughly speaking, our proof is based on the crucial observation that, the histogram of the $\ell_1$-norm of columns in the deterministic embedding $\Pi$ should look like that of a discretized standard Cauchy distribution.
I.e., there are at most $2^i$ columns in $\Pi$ with $\ell_1$-norm larger than $\Theta((n / d)2^{-i})$.
This is because if we choose a matrix $A \in \mathbb{R}^{n \times d}$ such that each column contains $(n / d)2^{-i}$ non-zero entries at random positions and all these $(n / d)2^{-i}$ non-zero entries are i.i.d. sampled from the standard Gaussian distribution $\normG$, then for each column in $A$, the $\ell_1$-norm of that column is $\Theta((n / d)2^{-i})$ with constant probability.
On the other hand, if the embedding $\Pi$ contains more than $2^i$ columns with $\ell_1$-norm larger than $\Theta((n / d)2^{-i})$, then with constant probability, there exists some $i \in [n]$ and $j \in [d]$ such that $A_{i, j} \sim \normG$ and the $i$-th column of $\Pi$ has $\ell_1$-norm larger than $\Theta((n / d)2^{-i})$.
In that case, it can be shown that after projection by $\Pi$, the $j$-th column of $A$ has $\ell_1$-norm larger than $\Theta((n / d)2^{-i})$, which violates the condition in (\ref{equ:intro}).

In order to prove $\kappa = \widetilde{\Omega}(d)$, let $c \in \mathbb{R}^{n}$ be a vector whose entries are all i.i.d. sampled from $\normG$.
With constant probability $\|c\|_1 = \Omega(n)$. On the other hand, we are able to show that the constraint we put on the histogram of the $\ell_1$-norm of columns in $\Pi$ implies that $\|\Pi c\|_1 = \widetilde{O}(n / d)$ and hence $\kappa = \widetilde{\Omega}(d)$.
A more refined analysis in Section \ref{sec:hardness} shows that $\kappa = \Omega(d \log^{-2} r)$ when $n \gg r$. 

In order to show that the dependence on $r$ in the lower bound is necessary, we construct an $\ell_1$ oblivious subspace embedding with $\exp(\exp(O(d)))$ rows and constant distortion.
The construction itself is the same as the dense construction in \cite{sohler2011subspace}.
Unlike previous approaches \cite{sohler2011subspace, meng2013low, woodruff2013subspace}, we do not use the existence of an Auerbach basis to prove the dilation bound.
Our analysis is based on tighter tail bounds for sums of absolute values of independent standard Cauchy (and also $p$-stable) random variables in Lemma \ref{lem:cauchy_upper} and \ref{lem:p_stable_lower}.
Let $\{X_i\}$ be $R = \exp(\exp(O(d)))$ independent standard Cauchy random variables. Based on the tighter tail bounds, it can be shown that with probability $1 - \exp(-\Omega(d))$,
$$
\sum_{i=1}^R |X_i| = \Theta(R \log R),
$$
which enables us to now apply a standard net argument to prove the constant distortion bound.
The formal analysis is given in Section \ref{sec:dependence_r}.
\paragraph{New $\ell_p$ oblivious subspace embeddings.}
For ease of notation, here we focus on $p = 1$.

Before getting into our results, we first review the construction in \cite{meng2013low} and its analysis.
The sparse Cauchy embedding in \cite{meng2013low} has $\widetilde{O}(d^5)$ rows.
In each column, there is a single non-zero entry which is sampled from the standard Cauchy distribution.
The $\widetilde{O}(d)$ dilation bound follows the standard approach \cite{sohler2011subspace} of using the existence of an Auerbach basis and upper tail bounds for dependent standard Cauchy random variables. 
The contraction bound is based on the technique of splitting coordinates, which was first proposed in \cite{clarkson2013low} to analyze the \textsf{CountSketch} embedding.
A coordinate is heavy if its $\ell_1$ leverage score is larger than $1/d$ and light otherwise.
For any vector $y = Ax$, if light coordinates contribute more to the $\ell_1$-norm of $y$, then standard concentration bounds and Cauchy lower tail bounds imply a constant distortion.
If heavy coordinates contribute more to the $\ell_1$-norm, since there will be at most $O(d^2)$ heavy coordinates and the embedding has $\Omega(d^4)$ rows, all the heavy coordinates will be perfectly hashed. 
An $\Omega\left(d^{-2}\right)$ contraction bound follows by setting up a global event saying that the absolute values of all of the $O(d^2)$ standard Cauchy random variables associated with the heavy coordinates are at least $\Omega\left(d^{-2}\right)$, which holds with constant probability. 

Although the dilation bound seems to be tight, the contraction bound can be improved.
Indeed, the $\ell_1$-norm of columns in the embedding of \cite{meng2013low} almost follows the histogram predicted by our lower bound argument, except for the lower tail part.
As predicted by our lower bound argument, for an embedding $\Pi$ which has the optimal $\kappa = \widetilde{O}(d)$ distortion, the $\ell_1$-norm of each column in $\Pi$ should be larger than a constant. 
On the other hand, the standard Cauchy distribution is heavy-tailed in both directions\footnote{This is also observed in \cite{woodruff2013subspace}, but the authors use exponential random variables there to remedy this issue instead of the idea of truncation that we use here.}. 
This leads to the idea of truncation, which is formalized in Section \ref{sec:truncating}.
The rough idea is that we make sure the absolute values of the standard Cauchy random variables are never smaller than a constant and thus the contraction bound can be improved to be a constant.
It is shown in Corollary \ref{lem:property_trunc_cauchy} that standard Cauchy random variables are still ``approximately 1-stable'' after truncation, which enables one to use Cauchy tail inequalities to analyze the dilation bound.
However, even though the distortion bound of this new embedding is nearly optimal, the number of rows is $\widetilde{O}(d^4)$, which seems difficult to improve.

Our alternate approach is still based on the technique of splitting coordinates. 
Unlike the approach in \cite{meng2013low} which is based on splitting coordinates according to the $\ell_1$ leverage scores, in this new approach, for any vector $y = Ax$, a coordinate $i$ is heavy if $|y_i| \ge \frac{1}{d^2}\|y\|_1$ and light otherwise.
When light coordinates contribute more to the $\ell_1$-norm of $y$, we show that the sparse Cauchy embedding in \cite{meng2013low} with only $O(d \log d)$ rows is already sufficient to deal with such vectors.
This is due to a tighter analysis based on negative association theory \cite{dubhashi1996balls} which also greatly simplifies the proof.
When heavy coordinates contribute more to the $\ell_1$-norm of $y$, the idea is to use known $\ell_2$ oblivious subspace embeddings.
The key observation is that when heavy coordinates contribute more to the $\ell_1$-norm, we have $\|y\|_2 \ge \Omega \left( \frac{1}{d} \right) \|y\|_1$ and thus any $\ell_2$ oblivious subspace embedding with constant distortion will also be an $\ell_1$ oblivious subspace embedding with $O(d)$ distortion.
See Section \ref{sec:l1_upper} for a formal analysis and Section \ref{sec:lp_upper} for how to generalize this idea to $\ell_p$-norms when $1 < p < 2$.

Our final embedding consists of two parts. 
The $\ell_2$ oblivious subspace embedding part could be the \textsf{CountSketch} embedding or the \textsf{OSNAP} embedding, which also provides a tradeoff between the number of non-zero entries per column and number of rows.
For the sparse Cauchy part, although it would be sufficient to prove the $O(d \log d)$ distortion bound as long as this part has $O(d \log d)$ rows, an analysis based on a tighter Cauchy lower tail bound in Lemma \ref{lem:p_stable_lower} shows that it is possible to further reduce the dilation to $O(d)$ by increasing the number of rows in this part.

Using this approach, the sparsest embedding we can construct has $O(d^2)$ rows and two non-zero entries per column.
We further show how to construct even sparser embeddings using random sampling.
Since we only use the sparse Cauchy part to deal with vectors in which light coordinates contribute most of the $\ell_1$-norm, even if we zero out each coordinate with probability $1 - \varepsilon$ for a small constant $\varepsilon$, the resulting vector will still have a sufficiently $\ell_1$-norm, with large enough probability.
Thus, if we zero out each standard Cauchy random variable in the sparse Cauchy part with probability $1 - \varepsilon$, the resulting embedding will still have the same distortion bound, up to a constant factor.
By doing so, there will be $1 + \varepsilon$ non-zero entries in expectation in each column of the new embedding.
This idea is formalized in Section \ref{sec:sampling}.

%% file: comparison.tex
\subsection{Comparison with Previous Work}
In order to compare our results with previous work, it is crucial to realize the difference between {\em oblivious} embeddings and {\em non-oblivious} embeddings.
An oblivious subspace embedding $\Pi$ is a universal distribution over $\mathbb{R}^{r \times n}$, which does not depend the given matrix $A \in \mathbb{R}^{n \times d}$.
A non-oblivious subspace embedding, however, is a distribution that possibly depends on the given matrix $A$.
Most known non-oblivious subspace embeddings involve importance sampling according to the {\em leverage scores} or {\em Lewis weights} of the rows, and so are inherently non-oblivious.
We refer the interested reader to \cite{mahoney2011randomized} for an excellent survey on leverage scores and \cite{li2013iterative, cohen2015uniform, cohen2015p} for recent developments on non-oblivious subspace embeddings.

Previous impossibility results for dimension reduction in $\ell_1$ \cite{lee2004embedding, brinkman2005impossibility, charikar2002dimension} are established by 
creating a set of $O(n)$ points in $\mathbb{R}^n$ and showing that any (non-oblivious) embedding on them incurs a large distortion.
In this paper, we focus on embedding a $d$-dimensional subspace of $\mathbb{R}^n$ into $\mathbb{R}^{\poly(d)}$ using oblivious embeddings.
We stress that $O(n)$ points in a $d$-dimensional subspace have a very different structure from $O(n)$ arbitrary points in $\mathbb{R}^n$.
Previous results \cite{cohen2015p} showed that any $d$-dimensional subspace in $\mathbb{R}^n$ can be embedded into $\mathbb{R}^{O(d (\log d) \varepsilon^{-2})}$ with $(1 + \varepsilon)$ distortion in $\ell_1$ using non-oblivious linear embeddings, where $\varepsilon > 0$ is an arbitrarily small constant.
Here the subspace structure is critically used, since Charikar and Sahai \cite{charikar2002dimension} showed that there exist $O(n)$ points such that any linear embedding $\mathbb{R}^n \to \mathbb{R}^d$ must incur a distortion of $\Omega(\sqrt{n / d})$, even for non-oblivious linear embeddings. 

Our hardness result in Theorem \ref{thm:main_lb} establishes a separation between oblivious and non-oblivious subspace embeddings in $\ell_p$ when $1 \le p< 2$.
This result suggests that in order to construct a subspace embedding with ($1 + \varepsilon$) distortion, it is essential to use non-oblivious subspace embeddings.

Although our main focus in this paper is to understand oblivious subspace embeddings, 
we remark that our technique for proving the hardness result in Theorem \ref{thm:main_lb} can also be applied to embed any $d$ points in $\mathbb{R}^n$ into $\mathbb{R}^{\poly(d)}$ in $\ell_p$ using oblivious linear embeddings, when $1 \le p < 2$.
In particular, it is possible to reproduce the result of \cite{charikar2002dimension} using our techniques, 
although in a weaker setting where the embeddings are oblivious. 

%% file: application.tex
\subsection{Applications of  Our Subspace Embeddings}
Using the sparse $\ell_p$ oblivious subspace embeddings in Theorem \ref{thm:main_lp_intro}, we obtain improvements to many related problems.
We list a few examples in this section.
\paragraph{$\ell_p$-regression in the distributed model.}
The $\ell_p$-regression problem in the distributed model is studied in \cite{woodruff2013subspace}, where there are $k$ clients $C_1, \ldots, C_k$ and one central server. 
Each client has a disjoint subset of the rows of a matrix $[A, b]$, where $A \in \mathbb{R}^{n \times d}$ and $b \in \mathbb{R}^n$. 
There is a $2$-way communication channel between the central server and the clients, and the goal of the server is to communicate with the $k$ clients to solve the $\ell_p$-regression problem $\argmin_{x} \|Ax-b\|_p$.

In \cite{woodruff2013subspace}, when $1 \le p < 2$, the authors devised an algorithm with total communication cost
$$
O\left(k d^{2+\eta} + d^5 \log^2 d + d^{3+p} \log(1/\varepsilon)/\varepsilon^2\right),
$$
where $\eta > 0$ is an arbitrarily small constant.
Using our new oblivious subspace embedding in Theorem \ref{thm:main_lp_intro}, the total communication cost can be further reduced to
$$
O\left(k d^2 \log d+ (d \log d)^{2 - p / 2}d^3 + d^{3+p} \log(1/\varepsilon)/\varepsilon^2\right),
$$
while the total running time of the whole system remains unchanged. 
This leads to a $\left(d \log d \right)^{p / 2}$ improvement to the second term of the communication cost.
\paragraph{$\ell_p$-regression in the streaming model.}
Using dense Cauchy embeddings and a sampling data structure from \cite{andoni2009efficient}, a single-pass streaming algorithm for $\ell_1$-regression $\argmin_{x} \|Ax - b\|_1$ was designed in \cite{sohler2011subspace}.
In order to get a $(1 + \varepsilon)$-approximate solution to the regression problem, the algorithm uses $\poly(d \varepsilon^{-1} \log n)$ bits of space,
where $A \in \mathbb{R}^{n \times d}$ and $b \in \mathbb{R}^n$.
The total running time of the algorithm, however, is $O(\nnz(A) \cdot d + \poly(d \varepsilon^{-1} \log n))$.

By replacing the dense Cauchy embedding with our new oblivious subspace embeddings in Theorem \ref{thm:main_lp_intro}, 
the total running time can be further improved to $\widetilde{O}(\nnz(A)) + \poly(d \varepsilon^{-1} \log n)$ while the space complexity remains unchanged. We note that using earlier sparse Cauchy embeddings \cite{meng2013low} would also give such a running time, but with a significantly worse $\poly(d \varepsilon^{-1} \log n)$ factor. 
The same approach can also be applied to design input-sparsity time algorithms for $\ell_p$-regression in the streaming model when $1 < p < 2$. 

\paragraph{Entrywise $\ell_p$ low rank approximation.}
Given a matrix $A \in \mathbb{R}^{n \times d}$ and approximation factor $\alpha$, the goal of the $\ell_1$-low rank approximation problem is to output a matrix $\widehat{A}$ for which
$$
\|A - \widehat{A}\|_1 \le \alpha \cdot \min_{\text{rank-$k$ matrices~}A'} \|A - A'\|_1,
$$
where $\|\cdot\|_1$ is the entrywise $\ell_1$-norm.

In \cite{song2017low}, the authors devised an algorithm that runs in $T = O\left(\nnz(A) + (n + d) \cdot \poly(k)\right)$ time to solve this problem, with $\alpha = \poly(k) \cdot \log d$.
The exact expression of the $\poly(k)$ factor in the approximation factor $\alpha$ and the running time $T$, depends on the number of rows $r$ and the distortion $\kappa$ of the $\ell_1$ oblivious subspace embedding used.
Both $\poly(k)$ factors can be directly improved by replacing the sparse Cauchy embedding \cite{meng2013low}, which is originally used in \cite{song2017low}, with our new oblivious subspace embeddings in Theorem \ref{thm:main_lp_intro}. 
This improvement also propagates to other problems considered in \cite{song2017low} such as $\ell_p$-low rank approximation, entrywise $\ell_p$-norm CUR decomposition and $\ell_p$-low rank approximation in distributed and streaming models.
\paragraph{Quantile Regression.}
Given a matrix $A \in \mathbb{R}^{n \times d}$ and $b \in \mathbb{R}^n$, the goal of {\em quantile regression} is to solve
$$
\argmin_{x} \rho_{\tau}(b - Ax),
$$
where $\rho_{\tau}(b - Ax) = \sum_{i=1}^n \rho_{\tau}((b - Ax)_i)$ and for any $z \in \mathbb{R}$,
$$
\rho_{\tau}(z) = \begin{cases}
\tau z & z \ge 0\\
(\tau - 1) z & z < 0
\end{cases}.
$$
Here $\tau$ is a parameter in $(0, 1)$.

An efficient algortihm to calculate a $(1 + \varepsilon)$-approximate solution to quantile regression was proposed in \cite{yang2013quantile}.
Using their approach, one can reduce a quantile regression instance of size $n \times d$ to a smaller instance of size $O(\poly(d)\varepsilon^{-2}\log(1 /  \varepsilon)) \times d$
in $O(\nnz(A)) + \poly(d)$ time.
By replacing the sparse Cauchy embedding, which is used in the conditioning step of their algorithm, with our new oblivious subspace embeddings in Theorem \ref{thm:main_lp_intro}, the $\poly(d)$ term in the running time can be directly improved. 

%% file: preliminary.tex
\newcommand{\constantupper}{U}
\newcommand{\constantlower}{L}

\section{Preliminaries}\label{sec:pre}
We use $\|\cdot \|_p$ to denote the $\ell_p$-norm of a vector or the entry-wise $\ell_p$-norm of a matrix.
The following lemma is a direct application of H\"older's inequality. 
\begin{lemma}\label{lem:inter_norm}
For any $x \in \mathbb{R}^n$ and $1 \le p \le q \le 2$, we have
$$
\|x\|_q \le \|x\|_p \le n^{1 / p - 1 / q} \|x\|_q.
$$
\end{lemma}
For $u \in \mathbb{R}^n$ and $1 \le a \le b \le n$, let $u_{a:b}$ denote the vector with $i$-th coordinate equal to $u_i$ when $i \in [a,b]$, and zero otherwise.
For a matrix $S \in \mathbb{R}^{n \times m}$, we use $S_{i, *}$ to denote the $i$-th row of $S$, and $S_{*, j}$ to denote the $j$-th column of $S$.

\begin{definition}
For $p \in [1, 2]$, a distribution over $r \times n$ matrices $\Pi$ is an $\ell_p$ oblivious subspace embedding, if for any fixed $A \in \mathbb{R}^{n \times d}$, 
$$
\Pr_{\Pi} \sqbrack{\|Ax\|_p \le \|\Pi Ax\|_p \le \kappa \|Ax\|_p, \forall x \in \mathbb{R}^{d} } \ge 0.99.
$$
Here $\kappa$ is the {\em distortion} of $\Pi$.
\end{definition}

Throughout the paper, we use $X \simeq Y$ to mean that $X$ and $Y$ have the same distribution.
We use $X \succeq Y$ to denote stochastic dominance, i.e., $X \succeq Y$ iff for any $t \in \mathbb{R}$,  $\Pr[X \ge t] \ge \Pr[Y \ge t]$.

\subsection{Stable Distribution} 
\begin{definition}[$p$-stable distribution]\label{def:p_stable}
A distribution $\mathcal{D}$ is $p$-stable if for any $n$ real numbers $a_1, a_2, \ldots, a_n$, we have
$$
\sum_{i=1}^n a_iX_i \simeq \left(\sum_{i=1}^n |a_i|^p\right)^{1 / p}X.
$$
Here $X_i$ are i.i.d. drawn from $\mathcal{D}$ and $X \sim \mathcal{D}$. 
\end{definition}
$p$-stable distributions exist for any $0 < p \le 2$ (see, e.g., \cite{nolan:2018}). 
We let $\mathcal{D}_p$ denote the $p$-stable distribution. 
It is also well known that the standard Cauchy distribution is $1$-stable 
and the standard Gaussian distribution $\normG$ is $2$-stable.

We use the following lemma due to Nolan \cite{nolan:2018}
\begin{lemma}[Theorem 1.12 in \cite{nolan:2018}]\label{lem:nolan_tail}
For $1 \le p < 2$, let $X_p \sim \mathcal{D}_p$.  As $t \to \infty$,
$$
\Pr[X_p > t] \sim c_p t^{-p},
$$
where $c_p > 0$ is a constant which depends only on $p$.
\end{lemma}
The following lemma is established in \cite{meng2013low} by using Lemma \ref{lem:nolan_tail}
\begin{lemma}[Lemma 8 in \cite{meng2013low}]\label{lem:mm13_dominance}
For $1 \le p < 2$, let $X_p \sim \mathcal{D}_p$. 
There exists a constant $\alpha_p$ such that
$$
\alpha_p |C| \succeq |X_p|^p,
$$
where C is a standard Cauchy random variable and $\alpha_p$ is a constant which depends only on $p$.
\end{lemma}

\subsection{Tail Inequalities} 
We use the following standard form of the 
Chernoff bound and Bernstein's inequality.
\begin{lemma}[Chernoff Bound] \label{lem:chernoff_bound}
Suppose $X_1, X_2, \ldots, X_n$ are independent random variables taking values in $[0, 1]$.
Let $X = \sum_{i=1}^n X_i$.

For any $\delta > 0$ we have
$$
\Pr\left[X >  (1 + \delta)\E[X]\right] \le \exp(-\delta^2\E[X] / 3),
$$
$$
\Pr\left[X <  (1 - \delta)\E[X]\right] \le \exp(-\delta^2\E[X] / 2).
$$

For $t > 2e\E[X]$ we have
$$
\Pr\left[X >  t\right] \le 2^{-t}.
$$

\end{lemma}
\begin{lemma}[Bernstein's inequality] \label{lem:bernstein}
Suppose $X_1, X_2, \ldots, X_n$ are independent random variables taking values in $[0, b]$. 
Let $X = \sum_{i=1}^n X_i$
and 
$\Var[X]= \sum_{i=1}^n \Var[X_i]$ be the variance of $X$. For any $t > 0$ we have
$$
\Pr[X > \E[X] + t] \le \exp\left(-\frac{t^2}{2 \Var[X] + 2bt/3}\right).
$$
\end{lemma}
The following Bernstein-type lower tail inequality is due to Maurer \cite{maurer2003bound}.
\begin{lemma}[\cite{maurer2003bound}]\label{lem:maurer}
Suppose $X_1, X_2, \ldots, X_n$ are independent positive random variables that satisfy $\E[X_i^2] < \infty$.
Let $X = \sum_{i=1}^n X_i$. For any $t > 0$ we have
$$
\Pr[X \le \E[X] - t] \le \exp\left(-\frac{t^2}{2 \sum_{i=1}^n \E[X_i^2]}\right).
$$
\end{lemma}

We use the following tail inequality of a Gaussian random vector. Its proof can be found in Appendix \ref{sec:proof}.

\begin{lemma}\label{lem:pnorm_weighted_gaussian}
Let $(a_1, a_2, \ldots, a_n)$ be a fixed vector. 
For $i \in [n]$, let $\{X_i\}$ be $n$ possibly dependent standard Gaussian random variables. 
For any $1 \le p \le 2$, we have
$$
\Pr\left[\left(\sum_{i=1}^n |a_iX_i|^p\right)^{1 / p} \in \left[C_p^{-1}\|a\|_p, C_p\|a\|_p\right] \right] \ge 0.99.
$$
Here $C_p > 1$ is an absolute constant which depends only on $p$.
\end{lemma}
The following upper tail inequality for dependent standard Cauchy random variables is established in \cite{meng2013low}.
\begin{lemma}[Lemma 3 in \cite{clarkson2016fast}]\label{lem:cauchy_dependent_uppertail}
For $i \in [n]$, let $\{X_i\}$ be $n$ possibly dependent standard Cauchy random variables and $\gamma_i > 0 $ with $\gamma = \sum_{i \in [n]} \gamma_i$. 
For any $t \ge 1$ and $n \ge 3$,
$$
\Pr\left[\sum_{i \in [n]} \gamma_i |X_i| > \gamma t\right] \le \frac{2\log(nt)}{t} .
$$
\end{lemma}
The following corollary is a direct implication of Lemma \ref{lem:cauchy_dependent_uppertail} and Lemma \ref{lem:mm13_dominance}.
\begin{corollary}\label{lem:p_stable_dependent_uppertail}
For $i \in [n]$, let $\{X_i\}$ be $n$ possibly dependent $p$-stable random variables and $\gamma_i > 0 $ with $\gamma = \sum_{i \in [n]} \gamma_i$. 
For any $t \ge 1$ and $n \ge 3$,
$$
\Pr\left[\sum_{i \in [n]} \gamma_i |X_i|^p > \alpha_p \gamma t\right] \le \frac{2\log(nt)}{t},
$$
where $\alpha_p$ is the constant in Lemma \ref{lem:mm13_dominance}.
\end{corollary}
For the sum of absolute values of {\em independent} standard Cauchy random variables, 
it is possible to prove an upper tail inequality stronger than that in Lemma \ref{lem:cauchy_dependent_uppertail}.
The proof of the following lemma can be found in Appendix \ref{sec:proof}.
\begin{lemma}\label{lem:cauchy_upper}
For $i \in [n]$, let $\{X_i\}$ be $n$ independent standard Cauchy random variables.
There exists a constant $\constantupper_1$, such that for any $n \ge 3$,
$$
\Pr\left[\sum_{i=1}^n |X_i| \le \constantupper_1 n \log n \right] \ge 1 - \frac{\log \log n}{\log n }.
$$
\end{lemma}
The following corollary is a direct implication of Lemma \ref{lem:cauchy_upper} and Lemma \ref{lem:mm13_dominance}.
\begin{corollary}\label{lem:p_stable_upper}
Suppose $1 \le p < 2$. For $i \in [n]$, let $\{X_i\}$ be $n$ independent $p$-stable random variables.
There exists a constant $\constantupper_p$ which depends only on $p$, such that for any $n \ge 3$,
$$
\Pr\left[\sum_{i=1}^n |X_i|^p \le \constantupper_p n \log n \right] \ge 1 - \frac{\log \log n}{\log n }.
$$
\end{corollary}
We use the following lower tail inequality for the sum of absolute values of independent standard Cauchy random variables, whose proof can be found in Appendix \ref{sec:proof}.
\begin{lemma}\label{lem:p_stable_lower}
Suppose $1 \le p < 2$. For $i \in [n]$, let $\{X_i\}$ be $n$ independent $p$-stable random variables.
There exists a constant $\constantlower_p$ which depends only on $p$, such that for sufficiently large $n$ and $T$,
$$
\Pr\left[ \sum_{i=1}^n |X_i|^p  \ge L_p n \log\left(\frac{n}{\log T}\right) \right] \ge 1 - \frac{1}{T}.
$$
\end{lemma}
\subsection{$\varepsilon$-nets} 
We use the standard $\varepsilon$-net construction of a subspace  in \cite{bourgain1989approximation}.
\begin{definition}
For any $1 \le p \le 2$, for a given $A \in \mathbb{R}^{n \times d}$, let $B = \{Ax \mid x \in \mathbb{R}^d, \|Ax\|_p = 1\}$. We say $\mathcal{N} \subseteq B$ is an $\varepsilon$-net of $B$ if for any $y \in B$, there exists a $\hat{y} \in \mathcal{N}$ such that $\|y - \hat{y}\|_p \le \varepsilon$.
\end{definition}
\begin{lemma}[\cite{bourgain1989approximation}]\label{lem:eps_net}
For a given $A \in \mathbb{R}^{n \times d}$, there exists an $\varepsilon$-net $\mathcal{N} \subseteq B = \{Ax \mid x \in \mathbb{R}^d, \|Ax\|_p = 1\}$ with size $|\mathcal{N}| \le (3 / \varepsilon)^d$.
\end{lemma}
\subsection{Known $\ell_2$ Oblivious Subspace Embeddings}
In \cite{clarkson2013low, meng2013low, nelson2013osnap, bourgain2015toward, cohen2016nearly}, a series of results on sparse $\ell_2$ oblivious subspace embedding are obtained.
\begin{lemma}[\textsf{CountSketch} \cite{clarkson2013low, meng2013low, nelson2013osnap}]\label{lem:cs}
There exists an $\ell_2$ oblivious subspace embedding over $O(d^2) \times n$ matrices $\Pi$, where each column of $\Pi$ has a single non-zero entry and the distortion $\kappa = 2$.
\end{lemma}
\begin{lemma}[\textsf{OSNAP} \cite{nelson2013osnap, cohen2016nearly}]\label{lem:osnap}
For any $B > 2$,  there exists an $\ell_2$ oblivious subspace embedding over $O\left(B \cdot d\log d\right) \times n$ matrices $\Pi$,
where each column of $\Pi$ has at most $O\left(\log_B d\right)$ non-zero entries and the distortion $\kappa = 2$.
\end{lemma}
For completeness we include the construction for \textsf{CountSketch} and \textsf{OSNAP} here.
In the \textsf{CountSketch} embedding, each column is chosen to have $s = 1$ non-zero entries chosen in a uniformly random location and the non-zero value is uniformly chosen in $\{-1,1\}$. 
In the \textsf{OSNAP} embedding, each column is chosen to have $s = O(\log_B d)$ non-zero entries in random locations, each equal to $\pm s^{-1 / 2}$ uniformly at random. 
All other entries in both embeddings are set to zero.

We need a few additional properties of the \textsf{CountSketch} embedding and the \textsf{OSNAP} embedding.
The following lemma is a direct calculation of the operator $\ell_1$-norm of the matrices stated above.
\begin{lemma}\label{lem:cs_osnap_l1}
For any $y \in \mathbb{R}^n$, 
\begin{itemize}
\item $\|\Pi y\|_1 \le \|y\|_1$ if $\Pi$ is sampled from the \textsf{CountSketch} embedding;
\item $\|\Pi y\|_1 \le  O(\log_B^{1 / 2} d)\|y\|_1$ if $\Pi$ is sampled from the \textsf{OSNAP} embedding.
\end{itemize}
\end{lemma}
The following lemma deals with the $\ell_p$-norm of a vector and its $\ell_p$-norm after projection using \textsf{CountSketch} or \textsf{OSNAP}.
Its proof can be found in Appendix \ref{sec:proof}.
\begin{lemma}\label{lem:cs_osnap_lp}
For any $y \in \mathbb{R}^n$ and sufficiently large $\omega$, with probability $1 - \exp(\Omega(\omega d \log d))$, 
\begin{itemize}
\item $\|\Pi \left(y_{1:d^2} \right) \|_p \le  (\omega d \log d)^{1 - 1 / p}\|y\|_p$ if $\Pi$ is sampled from the \textsf{CountSketch} embedding;
\item $\|\Pi \left(y_{1:d^2} \right) \|_p \le  (O(\log_Bd))^{1 / p - 1 / 2} (\omega d \log d)^{1 - 1 / p} \|y\|_p$ if $\Pi$ is sampled from the \textsf{OSNAP} embedding.
\end{itemize}
\end{lemma}
\subsection{Well-conditioned Bases}
We recall the definition and some existential results on well-coditioned matrices with respect to $\ell_p$-norms.
\begin{definition}[$(\alpha, \beta, p)$-well-conditioning \cite{dasgupta2009sampling}]
For a given matrix $U \in \mathbb{R}^{n\times d}$ and $p \in [1, 2]$, let $q$ be the dual norm of $p$, i.e., $1 / p + 1 / q = 1$.
We say $U$ is $(\alpha, \beta, p)$-well-conditioned if (i) $\|U\|_p \le \alpha$ and (ii) $\|x\|_q \le \beta \|Ux\|_p$ for any $x \in \mathbb{R}^d$.
\end{definition}
\begin{lemma}[Auerbach basis \cite{auerbach1930area}]
For any matrix $A$, there exists a basis matrix $U$ of $A$ such that $U$ is $(d^{1/p}, 1, p)$-well-conditioned.
\end{lemma}

%% file: hardness.tex
\newcommand{\scolumn}{$\mathsf{S}$-column}
\newcommand{\mcolumn}{\textsf{M}-column}
\newcommand{\dcolumn}{\textsf{D}-column}
\newcommand{\upperpnorm}{C_p^2}
\newcommand{\upperpnormlevel}{C_p^2}
\section{Hardness Result}
\subsection{The Lower Bound}\label{sec:hardness}
The goal of this section is to prove Theorem \ref{thm:main_lb}. We restate it here for convenience.

\noindent\textbf{Theorem~\ref{thm:main_lb}.} (restated)
\textit{For $1 \le p < 2$, if a distribution over $r \times n$ matrices $\Pi$ is an $\ell_p$ oblivious subspace embedding, then the distortion
$$
\kappa = \Omega \left(\frac{1}{\left(\frac{1}{d}\right)^{1 / p} \cdot \log^{2/p}r + \left(\frac{r}{n}\right)^{1 / p - 1 / 2}}\right).
$$
}

By Yao's minimax principle \cite{yao1977probabilistic}, it suffices to show that there exists a hard distribution $\hardA$ over $\mathbb{R}^{n \times d}$ such that 
for any $\Pi \in \mathbb{R}^{r \times n}$, if
\begin{equation}\label{equ:lp_ose}
\Pr_{A \sim \hardA} \sqbrack{\|Ax\|_p \le \|\Pi Ax\|_p \le \kappa \|Ax\|_p, \forall x \in \mathbb{R}^{d} } \ge 0.99,
\end{equation}
then 
$$
\kappa = \Omega \left(\frac{1}{\left(\frac{1}{d}\right)^{1 / p} \cdot \log^{2/p}r + \left(\frac{r}{n}\right)^{1 / p - 1 / 2}}\right).
$$

The columns in our construction of $\hardA$ consist of three parts:
\begin{itemize}
\item The first column is a vector where all the $n$ entries are i.i.d. standard Gaussian random variables.
We call this column the \dcolumn. 
\item
For the next $d / 4$ columns, each column has $4n / d$ non-zero entries, where all these non-zero entries are i.i.d. standard Gaussian random variables.
The indices of the $4n/d$ non-zero entries of the $i$-th column are $(4n/d)\cdot(i - 1) + 1, (4n/d)\cdot(i - 1) + 2, \ldots, (4n/d)\cdot i$.
We call each such column an \mcolumn.
\item We divide the next $d / 2$ columns into $\log(n / d)$ {\em blocks}, 
where each block contains $\frac{d}{2 \log(n / d)}$ columns.
For $0 \le i < \log(n / d)$, columns in the $i$-th block contain $2^{i + 1}$ non-zero entries and all of these non-zero entries are i.i.d. standard Gaussian random variables.
For two different columns in the same block, the sets of indices of non-zero entries are disjoint. 
For the $\frac{d}{2 \log(n / d)}$ columns in the $i$-th block, 
the indices of the $\frac{d}{2 \log(n / d)} \cdot 2^{i + 1} = \frac{d}{\log(n/d)} 2^i$ non-zero entries are sampled from $\{1, 2, \ldots, n\}$ without replacement.
We call each such column an \scolumn.
\end{itemize}

All entries in other columns are zero.
This finishes our construction of $\hardA$.

The following lemma is a direct implication of Lemma \ref{lem:pnorm_weighted_gaussian} and our construction. 
\begin{lemma} \label{lem:norm_of_hardA} 
For each column $c$ in $\hardA$, with probability at least $0.99$, the following holds:
\begin{enumerate}
\item If $c$ is an \scolumn~in the $i$-th block, then $\|c\|_p \le C_p 2^{(i + 1) / p} $.
\item If $c$ is an \mcolumn, then $\|c\|_p \le C_p (4n/d)^{1 / p}  $.
\item If $c$ is a \dcolumn, then $\|c\|_p \ge C_p^{-1} n^{1/ p}$. 
\end{enumerate}
Here $C_p$ is the constant in Lemma \ref{lem:pnorm_weighted_gaussian}.
\end{lemma} 

\begin{lemma}\label{lem:imp_of_m_column}
For any matrix $\Pi \in \mathbb{R}^{r \times n}$ which satisfies the condition in (\ref{equ:lp_ose}), 
the $\ell_p$-norm of each column of $\Pi$ is at most $\upperpnorm \kappa (4n / d)^{1 / p} $, 
where $C_p$ is the constant in Lemma \ref{lem:pnorm_weighted_gaussian}.
\end{lemma}
\begin{proof}
Suppose for contradiction that there exists an $i \in [n]$ for which the $i$-th column of $\Pi$ has $\ell_p$-norm larger than $C_p^2 \kappa (4n/d)^{1 / p} $.
Consider the vector $\mathsf{M}_j$, which is the $j$-th \mcolumn, whose $i$-th entry is a standard Gaussian random variable,
i.e., $(4n / d) \cdot (j - 1) + 1 \le i \le (4n / d) \cdot j$.
We first show that with probability at least $0.99$, $\|\Pi\mathsf{M}_j\|_p > C_p\kappa (4n/d)^{1 / p} $.
According to the $2$-stability of the standard Gaussian distribution, for any $k \in [r]$,
$$
(\Pi\mathsf{M}_j)_k \sim \left( \sum_{l=(4n/d) \cdot (j - 1) + 1}^{(4n / d) \cdot j}  \Pi_{k, l}^2 \right)^{1 / 2} \normG.
$$
Since
$$
\left( \sum_{l=(4n/d) \cdot (j - 1) + 1}^{(4n / d) \cdot j}  \Pi_{k, l}^2 \right)^{1 / 2} \ge \Pi_{k, i}, 
$$
according to Lemma \ref{lem:pnorm_weighted_gaussian}, with probability at least $0.99$,
$$
\|\Pi\mathsf{M}_j\|_p \ge C_p^{-1}\|\Pi_{*, i}\|_p > C_p \kappa (4n/d)^{1 / p} .
$$
According to Lemma \ref{lem:norm_of_hardA}, with probability at least $0.99$,
$$
\|\mathsf{M}_j\|_p \le C_p (4n/d)^{1 / p},
$$
which implies the condition in (\ref{equ:lp_ose}) is violated.
\end{proof}
\begin{lemma}\label{lem:imp_of_s_column}
For any matrix $\Pi \in \mathbb{R}^{r \times n}$ which satisfies the condition in (\ref{equ:lp_ose}), 
for any $0 \le i < \log(n/d)$, 
the number of columns in $\Pi$ with $\ell_p$-norm larger than 
$\upperpnormlevel \kappa 2^{(i + 1) / p}$ 
is at most 
$\frac{n \log (n / d)}{d} 2^{-i}$,
where $C_p$ is the constant in Lemma \ref{lem:pnorm_weighted_gaussian}.
\end{lemma}
\begin{proof}
Suppose for contradiction that for some $0 \le i < \log(n / d)$, 
the number of columns in $\Pi$ with $\ell_p$-norm larger than 
$C_p^2\kappa 2^{(i + 1) / p}$
is larger than
$\frac{n \log (n / d)}{d} 2^{-i}$.
Let $\mathsf{\Pi}^{1}, \mathsf{\Pi}^{2}, \ldots, \mathsf{\Pi}^{d\log^{-1}(n/d) / 2}$ be the $d\log^{-1}(n/d) / 2$ \scolumn~in the $i$-th block.
With probability at least $1 - \left(1 - \frac{\log(n/d)}{d} 2^{-i}\right)^{\frac{d}{\log(n / d)} 2^i} \ge 1 - 1 / e$, 
there exists a $j \in [d\log^{-1}(n/d) / 2]$ and $l \in [n]$ such that 
(i) $\|\Pi_{*, l}\|_p \ge C_p^2\kappa 2^{(i + 1) / p}$ and 
(ii) $\mathsf{\Pi}^{j}_{l}$ is a standard Gaussian random variable. 
According to Lemma \ref{lem:norm_of_hardA}, with probability at least $0.99$,
$\|\mathsf{\Pi}^{j}\|_p \le C_p 2^{(i + 1) / p}$.
Now we show that with probability at least $0.99$, 
$\|\Pi \mathsf{\Pi}^{j} \|_p \ge C_p^{-1} \|\Pi_{*, l}\|_p > C_p \kappa 2^{(i + 1) / p}$.
Suppose $P  \subseteq [n]$ is the set of indices at which $\mathsf{\Pi}^{j}$ contains a standard Gaussian random variable.
We know that $l \in P$.
Thus, due to the $2$-stability of the standard Gaussian distribution, for any $k \in [r]$,
$$
(\Pi \mathsf{\Pi}^{j})_k \sim \left(\sum_{m \in P} \Pi_{k, m}^2\right)^{1 / 2} \normG.
$$
Since
$$
\left(\sum_{m \in P} \Pi_{k, m}^2\right)^{1 / 2} \ge \Pi_{k, l},
$$
according to Lemma \ref{lem:pnorm_weighted_gaussian}, 
with probability at least $0.99$, 
$\|\Pi \mathsf{\Pi}^{j}\|_p 
\ge C_p^{-1} \|\Pi_{*, l}\|_p > C_p \kappa 2^{(i + 1) / p}
\ge \kappa \|\mathsf{\Pi}^{j}\|_p$,
which implies the condition in (\ref{equ:lp_ose}) is violated.
\end{proof}
\begin{lemma}
For any matrix $\Pi \in \mathbb{R}^{r \times n}$ which satisfies the condition in (\ref{equ:lp_ose}), we have
$$
\left(\sum_{i=1}^{r} \|\Pi_{i, *}\|_2^p \right)^{1 / p} = O\left(\kappa (n / d)^{1 / p} \log^{2 / p}(n / d) +  \kappa r^{\frac{1}{p} - \frac{1}{2}}  \sqrt{n}\right).
$$
\end{lemma}
\begin{proof}
We partition the columns of $\Pi$ into two parts. 
We let $\Pi^L$ be the submatrix of $\Pi$ formed by columns with $\ell_p$-norm at most $2^{1 / p} C_p^2 \kappa$
and $\Pi^H$ be the submatrix formed by columns with $\ell_p$-norm larger than $2^{1 / p} C_p^2  \kappa$.
For $\Pi^H$, by Lemma \ref{lem:imp_of_m_column} and Lemma \ref{lem:imp_of_s_column} we have
\begin{align*}
&\left(\sum_{i=1}^{r} \|\Pi^H_{i, *}\|_2^p \right)^{1 / p} \\
 \le &\left(\sum_{i=1}^{r} \|\Pi^H_{i, *}\|_p^p \right)^{1 / p}
 =\|\Pi^H\|_p
=\left(\sum_{i} \|\Pi^H_{*, i}\|_p^p \right)^{1 / p}\\
\le & 
\left( 
\sum_{i=0}^{\log(n / d) - 1} (\upperpnormlevel \kappa)^p \cdot 2^{i + 2} \cdot \frac{n \log (n / d)}{2^id} + (\upperpnormlevel \kappa)^p \cdot (4n / d) \cdot 2 \log(n / d)
\right)^{1 / p} \\
= & O\left(\kappa (n / d)^{1 / p} \log^{2 / p }(n / d)\right).
\end{align*}
For $\Pi^L$, since all the columns have $\ell_p$-norm at most $2^{1 / p} C_p^2 \kappa$, we have
\begin{align*}
&\left(\sum_{i=1}^{r} \|\Pi^L_{i, *}\|_2^p \right)^{1 / p} \\
 \le &r^{\frac{1}{p} - \frac{1}{2}}\left(\sum_{i=1}^{r} \|\Pi^L_{i, *}\|_2^2 \right)^{1 / 2}
 =r^{\frac{1}{p} - \frac{1}{2}} \|\Pi^L\|_2
= r^{\frac{1}{p} - \frac{1}{2}}\left(\sum_{i} \|\Pi^L_{*, i}\|_2^2 \right)^{1 / 2}\\
\le &r^{\frac{1}{p} - \frac{1}{2}}\left(\sum_{i} \|\Pi^L_{*, i}\|_p^2 \right)^{1 / 2} 
= O\left( \kappa r^{\frac{1}{p} - \frac{1}{2}}  \sqrt{n}\right),
\end{align*}
where the first inequality follows from Lemma \ref{lem:inter_norm} and the last equality follows from the fact that $\Pi^L$ has at most $n$ columns.

Notice that for any $1 \le i \le r$, $\|\Pi_{i, *}\|_2 \le \|\Pi^H_{i, *}\|_2 + \|\Pi^L_{i, *}\|_2$, which implies
$$
\left(\sum_{i=1}^{r} \|\Pi_{i, *}\|_2^p \right)^{1 / p} \le  \left(\sum_{i=1}^{r} \|\Pi^H_{i, *}\|_2^p \right)^{1 / p} + \left(\sum_{i=1}^{r} \|\Pi^L_{i, *}\|_2^p \right)^{1 / p} = O\left(\kappa (n / d)^{1 / p} \log^{2 / p}(n / d) +  \kappa r^{\frac{1}{p} - \frac{1}{2}}  \sqrt{n}\right).$$
\end{proof}
Now consider the vector $\mathsf{D}$, which is the \dcolumn~in $\hardA$.
According to Lemma \ref{lem:norm_of_hardA}, with probability at least $0.99$, $\|\mathsf{D}\|_p = \Omega(n^{1 / p})$. 
Due to the $2$-stability of the standard Gaussian distribution,
$$
(\Pi\mathsf{D})_i \sim \|\Pi_{i, *}\|_2 \normG.
$$
According to Lemma \ref{lem:pnorm_weighted_gaussian}, with probability at least $0.99$,
$$
\|\Pi\mathsf{D}\|_p = O\left(\sum_{i=1}^{r} \|\Pi_{i, *}\|_2^p \right)^{1 / p} = O\left(\kappa (n / d)^{1 / p} \log^{2/p}(n / d) +  \kappa r^{\frac{1}{p} - \frac{1}{2}}  \sqrt{n}\right).
$$
According to the condition in (\ref{equ:lp_ose}), we have
$$
\Omega(n^{1 / p}) = \|\mathsf{D}\|_p \le \|\Pi\mathsf{D}\|_p = O\left(\kappa (n / d)^{1 / p} \log^{2/p}(n / d) +  \kappa r^{\frac{1}{p} - \frac{1}{2}}  \sqrt{n}\right).
$$
which implies
$$
\kappa = \Omega \left(\frac{1}{\left(\frac{1}{d}\right)^{1 / p} \cdot \log^{2/p}(n / d) + \left(\frac{r}{n}\right)^{1 / p - 1 / 2}}\right).
$$
Now we show that the lower bound can be further improved to 
\begin{equation}\label{equ:lowerbound_weak}
\kappa = \Omega \left(\frac{1}{\left(\frac{1}{d}\right)^{1 / p} \cdot \log^{2/p}r + \left(\frac{r}{n}\right)^{1 / p - 1 / 2}}\right).
\end{equation}
We first note that $r$ should be at least $d$, otherwise if we take a full-rank matrix $A \in \mathbb{R}^{n \times d}$, $\text{rank}(\Pi A) < d = \text{rank}(A)$, 
which means we can find a non-zero vector $y = Ax$ in the column space of $A$ and $\Pi y = 0$, which implies the distortion $\kappa$ is not finite.

When $n \le rd^{2 / (2 - p)}$, $\log^{2/p}(n / d) = O(\log^{2/p}r)$, which means the lower bound in (\ref{equ:lowerbound_weak}) holds.
When $n > rd^{2 / (2 - p)}$, we repeat the argument above but only consider the first $rd^{2 / (2 - p)}$ columns of $\Pi$. By doing so we get a lower bound of 
$$
\kappa = \Omega \left(\frac{1}{\left(\frac{1}{d}\right)^{1 / p} \cdot \log^{2/p}(rd^{2 / (2 - p) - 1}) + \left(\frac{1}{d^{2 / (2 - p)}}\right)^{1 / p - 1 / 2}}\right) = \Omega(d^{1 / p} / \log^{2/p}(r)),
$$
which is always stronger than the lower bound of 
$$
\kappa = \Omega \left(\frac{1}{\left(\frac{1}{d}\right)^{1 / p} \cdot \log^{2/p}r + \left(\frac{r}{n}\right)^{1 / p - 1 / 2}}\right).
$$

%% file: discussion_hardness.tex
\newcommand{\numberofrows}{r}
\newcommand{\epsnet}{24 \left(\constantupper_p\constantlower_p^{-1}\right)^{1 / p}}
\subsection{Necessity of Dependence on  $r$}\label{sec:dependence_r}
The goal of this section is to prove Theorem \ref{thm:dependence_r}.
\begin{theorem}\label{thm:dependence_r}
Let $\numberofrows = \exp\left(4 \cdot 10^4 \cdot \left(\epsnet\right)^{2d}\right)$, 
where $\constantupper_p$ and $\constantlower_p$ are the constants in Corollary \ref{lem:p_stable_upper} and Lemma \ref{lem:p_stable_lower}, respectively. 
For $1 \le p < 2$, there exists an $\ell_p$ oblivious subspace embedding
over $\numberofrows \times n$ matrices $\Pi$, 
where the distortion $\kappa$ is a constant which depends only on $p$.
\end{theorem}

Our construction for the embedding in Theorem \ref{thm:dependence_r} is actually the same as the dense $p$-stable embedding in \cite{sohler2011subspace} (for $p = 1$) and Theorem 6 in \cite{meng2013low} (for $1 < p < 2$),
whose entries are i.i.d. sampled from the scaled $p$-stable distribution $ (\numberofrows \log \numberofrows) ^{-1 / p}\mathcal{D}_p$.

For any given matrix $A \in \mathbb{R}^{n \times d}$ and any $x \in \mathbb{R}^{d}$, we show that 
$$
\Pr_{\Pi}\left[ \left(\constantlower_p / 2\right)^{1 / p} \|Ax\|_p 
\le 
\|\Pi Ax\|_p 
\le 
\constantupper_p^{1 / p} \|Ax\|_p\right] 
\ge 
1 - 10^{-2} \left(24\left(\constantupper_p\constantlower_p^{-1}\right)^{1 / p}\right)^{-d}.
$$

According to the definition of the $p$-stable distribution in Definition \ref{def:p_stable}, for any $i \in [\numberofrows]$, 
$$
(\Pi Ax)_i \sim \left(\numberofrows \log \numberofrows \right)^{-1/ p} \|Ax\|_p \mathcal{D}_p.
$$
Since the entries in $\Pi$ are independent, the entries in the vector $\Pi Ax$ are also independent.  
Thus according to Corollary \ref{lem:p_stable_upper}, with probability at least $1 - \frac{\log \log \numberofrows} {\log \numberofrows} \ge  1 - 200^{-1} \left(\epsnet\right)^{-d}$, we have
$$
\|\Pi Ax\|_p^p \le \constantupper_p (\numberofrows \log \numberofrows)^{-1} \|Ax\|_p^p \cdot \numberofrows \log \numberofrows =  \constantupper_p  \|Ax\|_p^p,
$$
which implies
$$
\|\Pi Ax\|_p \le  \constantupper_p^{1 / p}  \|Ax\|_p.
$$
On the other hand, according to Lemma \ref{lem:p_stable_lower}, by setting $T = 200 \left(\epsnet\right)^d$, with probability at least $1 - 1 / T = 1 -  200^{-1}\left(\epsnet\right)^{-d}$, we have
$$
\|\Pi Ax\|_p^p \ge  \constantlower_p (\numberofrows \log \numberofrows)^{-1}\|Ax\|_p^p \cdot \numberofrows \log \frac{\numberofrows}{\log T} \ge  \constantlower_p / 2  \|Ax\|_p^p,
$$
which implies
$$
\|\Pi Ax\|_p \ge  \left( \constantlower_p^{1 / p} /2\right)^{1 / p}  \|Ax\|_p.
$$
It follows by a union bound that for any $x \in \mathbb{R}^d$,
$$
\Pr_{\Pi}\left[ \left(\constantlower_p / 2\right)^{1 / p} \|Ax\|_p 
\le 
\|\Pi Ax\|_p 
\le 
\constantupper_p^{1 / p} \|Ax\|_p\right] \ge 1 - 10^{-2} \left(24\left(\constantupper_p\constantlower_p^{-1}\right)^{1 / p}\right)^{-d}.
$$

We build an $\varepsilon$-net $\mathcal{N} \subseteq B = \{Ax \mid x \in \mathbb{R}^d, \|Ax\|_p = 1\}$ by setting $1 / \varepsilon = 8\left ( \constantupper_p\constantlower_p^{-1} \right)^{1 / p}$. 
According to Lemma \ref{lem:eps_net}, $|\mathcal{N}| \le (3 / \varepsilon)^ d = \left(24\left(\constantupper_p\constantlower_p^{-1}\right)^{1 / p}\right)^d$.
Again by a union bound, with probability at least $0.99$, we have for any $y \in \mathcal{N}$, 
$$
\left(\constantlower_p / 2\right)^{1 / p} \|y\|_p 
\le 
\|\Pi y\|_p 
\le 
\constantupper_p^{1 / p} \|y\|_p.
$$
Condition on the event stated above. Now we show that for any $x \in \mathbb{R}^d$, 
$$
\left(\constantlower_p / 4\right)^{1 / p} \|Ax\|_p 
\le 
\|\Pi Ax\|_p 
\le 
2\constantupper_p^{1 / p} \|Ax\|_p.
$$
For any $x \in \mathbb{R}^d$, let $y = Ax$. 
By homogeneity we can assume $\|y\|_p = 1$.
We claim $y$ can be written as
$$
y = y^0 + y^1 + y^2 + \ldots,
$$
where for any $i \ge 0$ we have (i) $\frac{y^i}{\|y_i\|_p} \in \mathcal{N}$ and (ii) $\|y^i\|_p \le \varepsilon^i$.

According to the definition of an $\varepsilon$-net, there exists a vector $y^0 \in \mathcal{N}$ for which $\|y - y^0\|_p \le \varepsilon$ and $\|y^0\|_p = 1$.
If $y = y_0$ then we stop. 
Otherwise we consider the vector $\frac{y - y^0}{\|y-y^0\|_p}$. 
Again we can find a vector $\hat{y}^1 \in \mathcal{N}$ such that $\left\| \frac{y - y^0}{\|y-y^0\|_p} - \hat{y}^1\right\|_p \le \varepsilon$ and $\|\hat{y}^1\|_p = 1$.
Here we set $y^1 = \|y-y^0\|_p \cdot \hat{y}^1$ and continue this process inductively.

It follows that
$$
\|\Pi y\|_p \ge \|\Pi y^0\| - \sum_{i > 0} \|\Pi y^i\| 
\ge \left(\constantlower_p / 2 \right)^{1 / p} - \sum_{i > 0} \constantupper_p^{1 / p} \varepsilon^i 
\ge  \left(\constantlower_p / 2 \right)^{1 / p} -  2\constantupper_p^{1 / p} \varepsilon \ge \left(\constantlower_p / 4 \right)^{1 / p}
$$
and
$$
\|\Pi y\|_p \le \sum_{i \ge 0} \|\Pi y^i\| 
\le  \sum_{i \ge 0} \constantupper_p^{1 / p} \varepsilon^i 
\le 2 \constantupper_p^{1 / p}.
$$
Thus, $\Pi$ is a valid $\ell_p$ oblivious subspace embedding with $\kappa \le 2\left(4\constantupper_p \constantlower_p^{-1 }\right)^{1 / p}$, which is a constant that depends only on $p$.

%% file: l1.tex
\section{New Subspace Embeddings for $\ell_1$}\label{sec:l1_upper}
In this section, we present new {\em sparse} $\ell_1$ oblivious subspace embeddings with nearly-optimal distortion.
\begin{theorem}\label{thm:main_l1_cs}
For any given $A \in \mathbb{R}^{n \times d}$, let $U$ be a $(d, 1, 1)$-well-conditioned basis of $A$. 
There exists an $\ell_1$ oblivious subspace embedding over $O(d^2) \times n$ matrices $\Pi$ where each column of $\Pi$ has two non-zero entries and
with probability $0.99$,  for any $x \in \mathbb{R}^d$,
$$
\Omega(\log d) \|Ux\|_1 \le \|\Pi Ux\|_1 \le O(d \log d) \|Ux\|_1.
$$
\end{theorem}
\begin{theorem}\label{thm:main_l1_osnap}
For any given $A \in \mathbb{R}^{n \times d}$ and sufficiently large $B$, 
let $U$ be a $(d, 1, 1)$-well-conditioned basis of $A$. 
There exists an $\ell_1$ oblivious subspace embedding over $O(B \cdot d \log d) \times n$ matrices $\Pi$ where each column of $\Pi$ has $O(\log_B d)$ non-zero entries and
with probability $0.99$,  for any $x \in \mathbb{R}^d$,
$$
\Omega(\log B) \|Ux\|_1 \le \|\Pi Ux\|_1 \le O(d \log d )\|Ux\|_1.
$$
\end{theorem}
Our embedding for Theorem \ref{thm:main_l1_cs} and Theorem \ref{thm:main_l1_osnap} can be written as $\Pi = (\Pi_1, \Pi_2)^T$.
For Theorem \ref{thm:main_l1_cs}, $\Pi_1$ is sampled from the \textsf{CountSketch} embedding in Lemma \ref{lem:cs}, scaled by a $d \log d$ factor.
For Theorem \ref{thm:main_l1_osnap}, $\Pi_1$ is sampled from the \textsf{OSNAP} embedding in Lemma \ref{lem:osnap} with $O(B \cdot d \log d)$ rows, $O(\log_B d)$ non-zero entries per column, and scaled by a $d \log B$ factor.
Suppose $\Pi_1$ has $R_1$ rows. 
Let $R_2 = \min\{R_1, d^{1.1}\}$.
$\Pi_2$ can be written as $\Phi D : \mathbb{R}^n \to \mathbb{R}^{R_2}$ as follows:
\begin{itemize}
\item $h:[n] \to [R_2]$ is a random map so that for each $i \in [n]$ and $t \in [R_2]$, $h(i)=t$ with probability $1/R_2$.
\item $\Phi$ is an $R_2 \times n$ binary matrix with $\Phi_{h(i),i} = 1$ and all remaining entries $0$.
\item $D$ is an $n \times n$ random diagonal matrix where the diagonal entries are i.i.d. sampled from the standard Cauchy distribution.
\end{itemize}

It is immediate to see that the number of rows in $\Pi_2$ is at most that in $\Pi_1$.
Furthermore, $\Pi_2$ has a single non-zero entry per column.

In the remainder of this section, we prove the dilation bound in Section \ref{sec:dilation_l1}, and the contraction bound in Section \ref{sec:contraction_l1}.
In the analysis we will define three events $\mathcal{E}_1$, $\mathcal{E}_2$ and $\mathcal{E}_{3}$, which we will condition on later in the analysis. 
We will prove that each of these events holds with probability at least $0.999$. 
By a union bound, all of these events hold with probability at least $0.997$.
Thus, these conditions will not affect our overall failure probability by more than $0.003$.
\subsection{No Overestimation}\label{sec:dilation_l1}
Let $\mathcal{E}_1$ be the event that $\|\Pi_2 U\|\le O\left( d \log d\right)$. We first prove that $\mathcal{E}_1$ holds with probability at least $0.999$.
\begin{lemma}\label{lem:cauchy_norm}
$\mathcal{E}_1$ holds with probability at least $0.999$.
\end{lemma}
\begin{proof}
\begin{align*}
\|\Pi_2 U\|_1 = \sum_{i=1}^{R_2} \sum_{j=1}^{d} |(\Pi_2 U)_{i, j}| 
= \sum_{i=1}^{R_2} \sum_{j=1}^{d} \left |\sum_{k \mid h(k) = i} D_{k, k} U_{k, j}  \right|
\simeq  \sum_{i=1}^{R_2} \sum_{j=1}^{d} \left (\sum_{k \mid h(k) = i} |U_{k, j}| \right) \left|\hat{X}_{i, j}\right|.
\end{align*}
Here $\left\{\hat{X}_{i, j}\right\}$ are dependent standard Cauchy random variables. 
Since $U$ is a $(d, 1, 1)$-well-conditioned basis of $A$, we have
$$
\sum_{i=1}^{R_2} \sum_{j=1}^{d} \left( \sum_{k \mid h(k) = i} |U_{k, j}| \right) \le d.
$$
By Lemma \ref{lem:cauchy_dependent_uppertail} we have
$$
\Pr\left[\|\Pi_2 U\|_1 > td\right] \le \frac{2\log(R_2 t d)}{t}.
$$
Taking $t = \omega \log d$ where $\omega$ is a sufficiently large constant, we have
$$
\Pr\left[\|\Pi_2 U\|_1 \le td \right] \ge 0.999.
$$
\end{proof}
\begin{lemma}\label{lem:cauchy_dilation_l1}
Conditioned on $\mathcal{E}_1$, for any $x \in \mathbb{R}^d$, we have
$$
\|\Pi_2 Ux\|_1 \le O(d \log d) \|Ux\|_1.
$$
\end{lemma}
\begin{proof}
$$
\|\Pi_2 Ux\|_1 \le \|\Pi_2 U\|_1 \|x\|_{\infty} \le \|\Pi_2  U\|_1 \|Ux\|_1 \le  O( d \log d ) \|Ux\|_1.
$$
The first inequality follows from H\"older's inequality, and the second inequality follows from the definition of a $(d, 1, 1)$-well-conditioned basis.
\end{proof}

Since $\Pi_1$ is the \textsf{CountSketch} embedding scaled by a $d \log d$ factor, 
or the \textsf{OSNAP} embedding scaled by a $d \log B$ factor, 
the following lemma is a direct implication of Lemma \ref{lem:cs_osnap_l1}.
\begin{lemma}\label{lem:l2_dilation_l1}
For any $x \in \mathbb{R}^d$, we have
$$
\|\Pi_1 Ux\|_1 \le O(d \log d) \|Ux\|_1.
$$
\end{lemma}
Combining Lemma \ref{lem:cauchy_dilation_l1} and Lemma \ref{lem:l2_dilation_l1}, we can bound the overall dilation of our embedding.
\begin{lemma}\label{lem:dilation_l1}
Conditioned on $\mathcal{E}_1$, for any $x \in \mathbb{R}^d$, we have
$$
\|\Pi Ux\|_1 \le O(d\log d) \|Ux\|_1.
$$
\end{lemma}
\begin{proof}
$$
\|\Pi Ux\|_1 = \|\Pi_1 Ux\|_1 + \|\Pi_2 Ux\|_1 =  O(d \log d) \|Ux\|_1.
$$
\end{proof}
\subsection{No Underestimation}\label{sec:contraction_l1}
We let $\mathcal{E}_2$ be the event that for any $x \in \mathbb{R}^d$,
$$
d \log d \|Ux\|_2 \le \|\Pi_1 Ux\|_2 \le 2 d \log d \cdot \|Ux\|_2 \qquad \text{(for Theorem \ref{thm:main_l1_cs})}
$$
or
$$
d \log B\|Ux\|_2 \le \|\Pi_1 Ux\|_2 \le 2 d \log B \cdot \|Ux\|_2 \qquad \text{(for Theorem \ref{thm:main_l1_osnap})}.
$$

Since $\Pi_1$ is sampled from an $\ell_2$ oblivious subspace embedding with $\kappa = 2$ and scaled by a factor of $d \log d$ (for Theorem \ref{thm:main_l1_cs}) or $d \log B$ (for Theorem \ref{thm:main_l1_osnap}), $\mathcal{E}_2$ holds with probability at least $0.999$.

Without loss of generality we assume $|x_1| \ge |x_2| \ge |x_3| \ge \ldots \ge |x_n|$. 
Of course, this order is unknown and is not used by our embedding.

We first show that for any $y = Ux$, if we can find a ``heavy'' part inside $y$, then the scaled $\ell_2$ oblivious subspace embedding $\Pi_1$ also works well for $\ell_1$. 
Formally, we have the following lemma.
\begin{lemma}\label{lem:heavy_contration_l1}
Conditioned on $\mathcal{E}_2$, for any $x \in \mathbb{R}^d$, if $\|(Ux)_{1 : d^2}\|_1
\footnote{Recall that for $u \in \mathbb{R}^n$ and $1 \le a \le b \le n$, $u_{a:b}$ denotes the vector with $i$-th coordinate equal to $u_i$ when $i \in [a,b]$, and zero otherwise.} 
\ge 0.5 \|Ux\|_1$, then
\begin{itemize}
\item $\|\Pi_1Ux\|_1 \ge \Omega(\log d) \|Ux\|_1$ for Theorem \ref{thm:main_l1_cs};
\item $\|\Pi_1Ux\|_1 \ge \Omega(\log B) \|Ux\|_1$ for Theorem \ref{thm:main_l1_osnap}.
\end{itemize}
\end{lemma}
\begin{proof}
Notice that
$$
\|Ux\|_2 \ge  \|(Ux)_{1 : d^2}\|_2
 \ge  \frac{1}{d} \|(Ux)_{1 : d^2}\|_1 \ge \frac{1}{2d} \|Ux\|_1,
$$
where the second inequality follows from Lemma \ref{lem:inter_norm}.
Thus for Theorem \ref{thm:main_l1_cs}, $\|\Pi_1 Ux\|_1 \ge \|\Pi_1 Ux\|_2 \ge \Omega(\log d) \|Ux\|_1$ since $\Pi_1$ is sampled from an $\ell_2$ oblivious subspace embedding and scaled by a factor of $d \log d$.
For Theorem \ref{thm:main_l1_osnap}, $\|\Pi_1 Ux\|_1 \ge \|\Pi_1 Ux\|_2 \ge \Omega(\log B) \|Ux\|_1$ since $\Pi_1$ is sampled from an $\ell_2$ oblivious subspace embedding and scaled by a factor of $d \log B$.
\end{proof}
Now we analyze those vectors $Ux$ that do {\em not} contain a ``heavy'' part. We show that they can be handled by the $\Pi_2$ part of our embedding. 
\begin{lemma}\label{lem:light_contraction_l1}
For any $x \in \mathbb{R}^d$, if $\|(Ux)_{d^2 + 1:n}\|_1 \ge 0.5 \|Ux\|_1$, then with probability at least $1 - \exp(-32 d \log d)$, we have
\begin{itemize}
\item $\|\Pi_2Ux\|_1 \ge \Omega(\log d) \|Ux\|_1$ for Theorem \ref{thm:main_l1_cs};
\item $\|\Pi_2Ux\|_1 \ge \Omega(\log B) \|Ux\|_1$ for Theorem \ref{thm:main_l1_osnap}.
\end{itemize}
\end{lemma}
\begin{proof}
Let $y = Ux$. 
By homogeneity we assume $\|y\|_ 1= 1$.
According to the given condition we have $\|y_{d^2 + 1: n}\|_1 \ge 0.5$.
Notice that $\|y_{d^2 + 1:n}\|_{\infty} \le 1 / d^2$ since otherwise $\|y_{1:d^2}\|_1 > d^2 \cdot 1 / d^2 = 1$.

For $i \in [R_2]$, let $B_i = \sum_{d^2 < j \le n} B_{i, j}$ where
$$
B_{i, j} = \begin{cases}
|y_j| & \text{if $h(j) = i$}\\
0 & \text{otherwise}
\end{cases}.
$$
It follows that $\sum_{i=1}^{R_2} B_i = \|y_{d^2+1:n}\|_1 \ge 0.5$.

Since $\|y_{d^2 + 1:n}\|_{\infty} \le 1 / d^2$ and $1 / 2 \le \|y_{d^2 + 1:n}\|_{1} \le 1$, for any $i \in [R_2]$ and $j > d^2$ we have
$$
B_{i, j} \le \frac{1}{d^2}
$$
and
$$
\frac{1}{2R_2} \le \E[B_i]  \le \frac{1}{R_2}.
$$
Furthermore, by H\"older's inequality we have
$$
\Var[B_i] = \sum_{j=1}^n \Var[B_{i, j}] \le \frac{1}{R_2} \sum_{j=d^2 + 1}^n y_j^2 \le \frac{1}{R_2} \|y_{d^2 + 1:n}\|_{\infty} \cdot \|y_{d^2 + 1:n}\|_1 \le \frac{1}{R_2d^2}.
$$
Thus by Bernstein's inequality in Lemma \ref{lem:bernstein}, we have
\begin{equation}\label{equ:bernstein_l1}
\Pr[B_i \ge 1 / R_2 + t] \le \exp\left(-\frac{t^2}{\frac{2}{R_2d^2}+ \frac{2t}{3d^2}} \right).
\end{equation}

Let $t = d^{0.2} / R_2$.
Since $R_2 \le d^{1.1}$, by (\ref{equ:bernstein_l1}) we have
$$
\Pr[B_i > (d^{0.2} + 1)/R_2] \le \exp\left(-3d^{2.2} / 4R_2\right) \le \exp\left(-3d^{1.1}/4\right).
$$
By a union bound, with probability at least 
$$
1 - \exp\left(-3d^{1.1}/4\right) \cdot R_2 \ge 1 - \exp(-32d \log d) / 4,
$$
simultaneously for all $i \in [R_2]$ we have $B_i \le (d^{0.2}+ 1) / R_2$.

Let $t = 1 / R_2$.
Since $R_2 \le d^{1.1}$, by (\ref{equ:bernstein_l1}) we have
$$
\Pr[B_i > 2/R_2] \le \exp\left(-3d^2 / 8R_2\right) \le \exp\left(-3d^{0.9}/8\right).
$$

According to \cite{dubhashi1996balls}, $B_i$ are negatively associated, which implies for any $I \subseteq [R_2]$ we have
$$
\Pr[B_i \ge t_i, i \in I] \le \prod_{i \in I} \Pr[B_i \ge t_i].
$$
Thus, the probability that the number of $B_i$ which satisfy $B_i > 2 / R_2$ is larger than $d^{0.2}$, is at most
$$
\binom{R_2}{d^{0.2}} \exp(-3d^{0.9}/8 \cdot d^{0.2}) <\exp(-32 d\log d) / 4.
$$

It follows that with probability at least $1 - \exp(-32d \log d) / 2$, 
for any $i \in [R_2]$ we have $B_i \le (d^{0.2} + 1) / R_2$,
and the number of $B_i$ which satisfy $B_i > 2 / R_2$ is at most $d^{0.2}$.
In the rest of the proof we condition on this event.

Since $R_2 \ge d \log d$, 
$$
\sum_{i \in [R_2] \mid B_i \le 2 / R_2} B_i \ge 0.5 - d^{0.2} \cdot (d^{0.2} + 1) / R_2 \ge 1 / 4.
$$
Thus, the number of $B_i$ which satisfy $B_i \ge \frac{1}{8R_2}$ is at least $R_2 / 16$, since otherwise
$$
\sum_{i \in [R_2] \mid B_i \le 2 / R_2} < R_2 / 16 \cdot 2 / R_2 + R_2 \cdot \frac{1}{8R_2} = 1 / 4.
$$

Now consider $\Pi_2 y$.
According to the $1$-stability of the standard Cauchy distribution,
$$
\left|(\Pi_2 y)_i\right| \simeq \left( \sum_{j \in [n] \mid h(j) = i} |y_j| \right) \cdot |X_i|
$$
where $\left\{X_i\right\}$ are independent standard Cauchy random variables. 
Notice that conditioned on the event stated above, the number of $B_i$ which satisfy $B_i \ge \frac{1}{8R_2}$ is at least $R_2 / 16$.
Furthermore, for any $i \in [R_2]$, $\sum_{j \in [n] \mid h(j) = i} |y_j|  \ge \sum_{d^2 < j \le n \mid h(j) = i} |y_j|  = B_i$.
Thus, 
$$
\sum_{i=1}^{R_2} \left|(\Pi_2 y)_i\right| \succeq \sum_{i=1}^{R_2 / 16} \frac{1}{8R_2} |\overline{X}_i|,
$$
where $\left\{\overline{X}_i\right\}$ are independent standard Cauchy random variables. 

According to Lemma \ref{lem:p_stable_lower}, by setting $T = 2\exp(32 d \log d)$, with probability at least $1 - 1 / T$ we have
$$
\sum_{i=1}^{R_2} \left|(\Pi_2 y)_i\right| \ge \constantlower_1 \cdot R_2 / 16 \cdot \log(R_2 / (16\log T)) \cdot \frac{1}{8R_2} = \Omega(\log(R_2 / \log T)).
$$
Thus, 
for Theorem \ref{thm:main_l1_cs}, we have $\|\Pi_2 y\|_1 \ge \Omega(\log d)$ since $R_2 = d^{1.1}$ and $\log T = O(d \log d)$.
For Theorem \ref{thm:main_l1_osnap}, when $R_1 \le d^{1.1}$, we have $\|\Pi_2 y\|_1 \ge \Omega(\log B)$ since $R_2 = R_1 = O(B \cdot d \log d)$ and $\log T = O(d \log d)$.
When $R_1 > d^{1.1}$, we have $\|\Pi_2 y\|_1 \ge \Omega(\log d) = \Omega(\log B)$ since $R_2 = d^{1.1}$ and $\log T = O(d \log d)$.

\end{proof}
Set $\varepsilon = 1 / d^2$ and create an $\varepsilon$-net $\mathcal{N} \subseteq B = \{Ux \mid x \in \mathbb{R}^d\text{ and } \|Ux\|_1 = 1\}$.
According to Lemma \ref{lem:eps_net}, $|\mathcal{N}| \le (3d^2)^d$.
Let $\mathcal{E}_3$ be the event that for all $y \in \mathcal{N}$,  if $\|y_{d^2 + 1:n}\|_1 \ge 0.5$, then
$\|\Pi_2 y\|_1 \ge {\Omega}(\log d)\|y\|_1$ (for Theorem \ref{thm:main_l1_cs})
or
$\|\Pi_2 y\|_1 \ge {\Omega}(\log B)\|y\|_1$ (for Theorem \ref{thm:main_l1_osnap}). 

Now we show that $\mathcal{E}_3$ holds with constant probability. 
\begin{lemma}\label{lem:light_contration_all_l1}
$\mathcal{E}_3$ holds with probability at least $0.999$.
\end{lemma}
\begin{proof}
According to Lemma \ref{lem:light_contraction_l1}, by using a union bound, we have
$$
\Pr[\mathcal{E}_3 \text{ holds}] \ge 1 - |\mathcal{N}| \exp(-32d \log d)  > 0.999
$$
\end{proof}
We are now ready to prove the contraction bound.
\begin{lemma}
Conditioned on $\mathcal{E}_1$, $\mathcal{E}_2$ and $\mathcal{E}_3$, for all $x \in \mathbb{R}^d$, we have
\begin{itemize}
\item $\|\Pi_2Ux\|_1 \ge \Omega(\log d) \|Ux\|_1$ for Theorem \ref{thm:main_l1_cs};
\item $\|\Pi_2Ux\|_1 \ge \Omega(\log B) \|Ux\|_1$ for Theorem \ref{thm:main_l1_osnap}.
\end{itemize}
\end{lemma}
\begin{proof}
By homogeneity we can assume $\|Ux\|_1 = 1$.
According to Lemma \ref{lem:heavy_contration_l1}, conditioned on $\mathcal{E}_2$ and $\mathcal{E}_3$, for all $y \in \mathcal{N}$, 
we have
$\|\Pi_2Ux\|_1 \ge \Omega(\log d) \|Ux\|_1$ (for Theorem \ref{thm:main_l1_cs})
or
$\|\Pi_2Ux\|_1 \ge \Omega(\log B) \|Ux\|_1$ (for Theorem \ref{thm:main_l1_osnap}).
For any given $y = Ux$ where $\|y\|_1 = 1$, there exists some $\hat{y} \in \mathcal{N}$ for which $\|y - \hat{y}\|_1 \le \varepsilon = 1/ d^2$.
Thus, conditioned on $\mathcal{E}_1$, notice that both $\hat{y}$ and $y - \hat{y}$ are in the column space of $U$, so according to Lemma \ref{lem:dilation_l1},
we have
$$
\|\Pi y\|_1 \ge \|\Pi \hat{y}\|_1 - \|\Pi(y - \hat{y})\|_1 \ge \Omega(\log d) - 1/d^2 \cdot O(d \log d) = \Omega(\log d) \text{\qquad(for Theorem \ref{thm:main_l1_cs})}
$$
or
$$
\|\Pi y\|_1 \ge \|\Pi \hat{y}\|_1 - \|\Pi(y - \hat{y})\|_1 \ge \Omega(\log B) - 1/d^2 \cdot O(d \log d) = \Omega(\log B) \text{\qquad(for Theorem \ref{thm:main_l1_osnap})}.
$$

\end{proof}

%% file: lp.tex
\section{New Subspace Embeddings for $\ell_p$}\label{sec:lp_upper}
In this section, we show how to generalize the constructions in Section \ref{sec:l1_upper} to $\ell_p$-norms, for $1 < p < 2$.
\begin{theorem}\label{thm:main_lp_cs}
Suppose $1 \le p < 2$.
For any given $A \in \mathbb{R}^{n \times d}$, let $U$ be a $(d^{1 / p}, 1, p)$-well-conditioned basis of $A$. 
There exists an $\ell_p$ oblivious subspace embedding over $O(d^2) \times n$ matrices $\Pi$ where each column of a matrix drawn from $\Pi$ has two non-zero entries and
with probability $0.99$,  for any $x \in \mathbb{R}^d$,
$$
\Omega(1) \|Ux\|_p \le \|\Pi Ux\|_p \le O\left((d \log d)^{1 / p}\right) \|Ux\|_p.
$$
\end{theorem}
\begin{theorem}\label{thm:main_lp_osnap}
Suppose $1 \le p < 2$.
For any given $A \in \mathbb{R}^{n \times d}$ and sufficiently large $B$, 
let $U$ be a $(d^{1 / p}, 1, p)$-well-conditioned basis of $A$. 
There exists an $\ell_p$ oblivious subspace embedding over $O(B \cdot d \log d) \times n$ matrices $\Pi$ where each column of a matrix drawn from $\Pi$ has $O(\log_B d)$ non-zero entries and
with probability $0.99$,  for any $x \in \mathbb{R}^d$,
$$
\Omega(1) \|Ux\|_p \le \|\Pi Ux\|_p \le O\left((d \log d)^{1 / p} \right)\|Ux\|_p.
$$
\end{theorem}
Our embeddings for Theorem \ref{thm:main_lp_cs} and Theorem \ref{thm:main_lp_osnap} can be written as $\Pi = (\Pi_1, \Pi_2)^T$.
Similar to the constructions in Section \ref{sec:l1_upper}, 
for Theorem \ref{thm:main_lp_cs}, $\Pi_1$ is sampled from the \textsf{CountSketch} embedding in Lemma \ref{lem:cs}, scaled by a
$d^{2/p - 1}$
factor.
For Theorem \ref{thm:main_lp_osnap}, $\Pi_1$ is sampled from the \textsf{OSNAP} embedding in Lemma \ref{lem:osnap} with $O(B \cdot d \log d)$ rows and $O(\log_B d)$ non-zero entries per column and also scaled by a $d^{2/p - 1}$ factor.
The construction for $\Pi_2$ is almost the same as that for Theorem \ref{thm:main_l1_cs} and \ref{thm:main_l1_osnap}, except for replacing the standard Cauchy random variables in the diagonal entries of $D$ with $p$-stable random variables.
Most parts of the proof for the distortion bound resemble that for Theorem \ref{thm:main_l1_cs} and \ref{thm:main_l1_osnap}.
We will omit similar proofs. 

The following lemma can be proved in the same way as Lemma \ref{lem:cauchy_norm} and Lemma \ref{lem:cauchy_dilation_l1},
except for replacing the upper tail inequality for standard Cauchy random variables in Lemma \ref{lem:cauchy_dependent_uppertail}
with that for $p$-stable random variables in Corollary \ref{lem:p_stable_dependent_uppertail},
and replacing the properties of a $(d, 1, 1)$-well-conditioned basis with those of a $(d^{1 / p}, 1, p)$-well-conditioned basis.
\begin{lemma}\label{lem:p_stable_dilation_lp}
Let $\mathcal{E}_1$ be the event that $\|\Pi_2 U\|_p \le O\left( (d \log d)^{1 / p}\right)$.
$\mathcal{E}_1$ holds with probability at least $0.999$.
Furthermore, conditioned on $\mathcal{E}_1$, for any $x \in \mathbb{R}^d$, we have
$$
\|\Pi_2 Ux\|_p \le O\left( (d \log d)^{1 / p}\right) \|Ux\|_p.
$$
\end{lemma}

Let $\mathcal{E}_2$ be the event that for any $x \in \mathbb{R}^d$,
$$
d^{2/p - 1}\|Ux\|_2 \le \|\Pi_1 Ux\|_2 \le 2d^{2/p - 1} \|Ux\|_2.
$$
Since $\Pi_1$ is sampled from an $\ell_2$ oblivious subspace embedding with $\kappa = 2$, and scaled by a factor of $d^{2/p - 1}$, $\mathcal{E}_2$ holds with probability at least $0.999$.

Without loss of generality we assume $|x_1| \ge |x_2| \ge |x_3| \ge \ldots \ge |x_n|$. 
Of course, this order is unknown and is not used by our embeddings.

\begin{lemma}\label{lem:l2_dilation_lp}
Conditioned on $\mathcal{E}_2$, for any $x \in \mathbb{R}^d$, we have
$$
\|\Pi_1 (Ux)_{d^2 + 1:n}\|_p \le O\left((d \log d)^{1 / p}\right) \|Ux\|_p.
$$
\end{lemma}
\begin{proof}
By homogeneity we can assume $\|Ux\|_p = 1$.
Notice that $\|(Ux)_{d^2 + 1 : n}\|_{\infty} \le d^{-2 / p}$ since otherwise $\|Ux\|_p \ge \|(Ux)_{1:d^2}\|_p > 1$.
By H\"older's inequality, 
$$
\|(Ux)_{d^2+1:n}\|_2 = \left(\sum_{i=d^2+1}^n (Ux)_i^2\right)^{1/2} \le \left( \sum_{i=d^2+1}^n |(Ux)_i|^p \cdot \max_{d^2+1 \le i \le n} |Ux_i|^{2 - p}\right)^{1 / 2} \le d^{1 - 2/p}.
$$
Thus,
\begin{align*}
&\|\Pi_1 (Ux)_{d^2 + 1:n}\|_p \le R_1^{1/ p - 1 /2} \|\Pi_1 (Ux)_{d^2 + 1:n}\|_2 \\
\le &O(d^{2 / p - 1}) \cdot 2 d^{2 / p - 1} \|(Ux)_{d^2+1 : n}\|_2 = O(d^{2 / p - 1}) = O\left((d \log d)^{1 / p} \right).
\end{align*}
Here the first inequality follows from Lemma \ref{lem:inter_norm} and the fact that $\Pi_1 U x$ has $R_1$ rows, the second inequality holds since $R_1 \le O(d^2)$ and $\mathcal{E}_2$ holds.
\end{proof}
\begin{lemma}\label{lem:heavy_contration_lp}
Conditioned on $\mathcal{E}_2$, for any $x \in \mathbb{R}^d$, if $\|(Ux)_{1 : d^2}\|_p
\ge 0.5 \|Ux\|_p$, then $\|\Pi_1Ux\|_p \ge \Omega(1) \|Ux\|_p$.
\end{lemma}
\begin{proof}
Notice that
$$
\|Ux\|_2 \ge  \|(Ux)_{1 : d^2}\|_2
 \ge d^{1 - 2 / p} \|(Ux)_{1 : d^2}\|_p \ge d^{1 - 2 / p} / 2 \|Ux\|_p,
$$
where the second inequality follows from Lemma \ref{lem:inter_norm} and the third inequality follows from the condition that $\|(Ux)_{1 : d^2}\|_p \ge 0.5 \|Ux\|_p$.
Thus, $\|\Pi_1 Ux\|_p \ge \|\Pi_1 Ux\|_2 \ge \Omega(1) \|Ux\|_p$ since $\Pi_1$ is sampled from an $\ell_2$ oblivious subspace embedding and scaled by a factor of $d^{2 / p - 1}$.
\end{proof}
The proof of the following lemma is almost identical to that of Lemma \ref{lem:light_contraction_l1}.
We omit the proof here.
\begin{lemma}\label{lem:light_contraction_lp}
For any $x \in \mathbb{R}^d$, if $\|(Ux)_{d^2 + 1:n}\|_p \ge 0.5 \|Ux\|_p$, then with probability at least $1 - \exp(-32 d \log d)$, we have $\|\Pi_2Ux\|_p \ge \Omega(1) \|Ux\|_p$.
\end{lemma}
Set $\varepsilon = 1 / d^2$ and create an $\varepsilon$-net $\mathcal{N} \subseteq B = \{Ux \mid x \in \mathbb{R}^d\text{ and } \|Ux\|_p = 1\}$.
According to Lemma \ref{lem:eps_net}, $|\mathcal{N}| \le (3d^2)^d$.

Let $\mathcal{E}_3$ be the event that for all $y \in \mathcal{N}$,
\begin{enumerate}
\item if $\|y_{d^2 + 1:n}\|_p \ge 0.5 \|y\|_p$, then $\|\Pi_2 y\|_p \ge {\Omega}(1)\|y\|_p$;
\item $\|\Pi_1 \left( y_{1:d^2}\right) \|_p \le O\left((d \log d)^{1 / p}\right) \|y\|_p$.
\end{enumerate}
\begin{lemma}\label{lem:light_contration_all_lp}
$\mathcal{E}_3$ holds with probability at least $0.999$.
\end{lemma}
\begin{proof}
Notice that $\Pi_1$ is sampled from the \textsf{CountSketch} embedding or the \textsf{OSNAP} embedding and scaled by a $d^{2 / p - 1} $ factor. According to Lemma \ref{lem:cs_osnap_lp},
by setting $\omega$ sufficiently large,
with probability $1 - \exp(-32d \log d)$
we have $\|\Pi_1 \left( y_{1:d^2}\right) \|_p \le d^{2 / p - 1} \left(O(\log_Bd)\right)^{1 / p - 1 / 2} (\omega d \log d)^{1 - 1 / p}  \|y\|_p = O\left((d \log d)^{1 / p}\right)\|y\|_p$.
Combining this with Lemma \ref{lem:light_contraction_lp} and a union bound, we have
$$
\Pr[\mathcal{E}_3 \text{ holds}] \ge 1 - 2  |\mathcal{N}| \exp(-32d \log d)  > 0.999
$$
\end{proof}
\begin{lemma}
Conditioned on $\mathcal{E}_1$, $\mathcal{E}_2$ and $\mathcal{E}_3$, for all $x \in \mathbb{R}^d$, we have
$$
\Omega(1) \|Ux\|_p \le \|\Pi Ux\|_p \le O\left((d \log d)^{1 / p}\right) \|Ux\|_p.
$$
\end{lemma}
\begin{proof}
For any $x \in \mathbb{R}^d$, let $y = Ux$. 
By homogeneity we can assume $\|y\|_p = 1$.
As in the proof of Theorem \ref{thm:dependence_r}, $y$ can be written as
$$
y = y^0 + y^1 + y^2 + \ldots,
$$
where for any $i \ge 0$ we have (i) $\frac{y^i}{\|y_i\|_p} \in \mathcal{N}$ and (ii) $\|y^i\|_p \le \varepsilon^i$.

It follows by Lemma \ref{lem:p_stable_dilation_lp}, Lemma \ref{lem:l2_dilation_lp} and Lemma \ref{lem:heavy_contration_lp} that 
$$
\|\Pi y\|_p \ge \|\Pi y^0\|_p - \sum_{i > 0} \|\Pi y^i\|_p
\ge \Omega(1) - \sum_{i > 0} O\left((d \log d)^{1 / p }\right) \varepsilon^i 
\ge \Omega(1) - O\left((d \log d)^{1 / p }\right) 2 \varepsilon
\ge \Omega(1)
$$
and
$$
\|\Pi y\|_p \le \sum_{i \ge 0} \|\Pi y^i\|_p 
\le  \sum_{i \ge 0} O\left((d \log d)^{1 / p}\right)\varepsilon^i 
\le O\left((d \log d)^{1 / p}\right).
$$

\end{proof}

%% file: sparsity.tex
\newcommand{\trunc}[2]{\mathsf{trunc}_{#1}(#2)}
\section{Subspace Embeddings with Improved Sparsity}\label{sec:improved_sparsity}
In this section, we present two approaches to constructing {\em sparser} $\ell_p$ oblivious subspace embeddings for $1 \le p < 2$.
In Section \ref{sec:sampling} we present our first approach based on random sampling, which yields an $\ell_p$ oblivious subspace embedding where each column of the embedding has at most two non-zero entries and $1 + \varepsilon$ non-zero entries in expectation, where the number of rows $r = O(d^2)$.
In Section \ref{sec:truncating}, we present another approach based on the construction in \cite{meng2013low} and a truncation argument, which yields an $\ell_p$ oblivious subspace embedding where each column of the embedding has a single non-zero entry,
at the cost of increasing the number of rows $r$ to $\widetilde{O}(d^4)$.
\subsection{Improved Sparsity Based on Random Sampling}\label{sec:sampling}
In this section, we show how to further improve the sparsity in the constructions of Theorem \ref{thm:main_l1_cs} and Theorem \ref{thm:main_lp_cs}.
\begin{theorem}\label{thm:main_l1_sparse}
For any given $A \in \mathbb{R}^{n \times d}$, let $U$ be a $(d, 1, 1)$-well-conditioned basis of $A$. 
For any constant $0 < \varepsilon < 1$, there exists an $\ell_1$ oblivious subspace embedding over $O(d^2) \times n$ matrices $\Pi$ 
where each column of $\Pi$ has at most two non-zero entries and $1 + \varepsilon$ non-zero entries in expectation,
such that
with probability $0.99$,  for any $x \in \mathbb{R}^d$,
$$
\Omega(\log d) \|Ux\|_1 \le \|\Pi Ux\|_1 \le O(d \log d) \|Ux\|_1.
$$
\end{theorem}
Our embedding for Theorem \ref{thm:main_l1_sparse} is almost identical to that for Theorem \ref{thm:main_l1_cs} except for the $\Pi_2$ part.
Recall that the $\Pi_2$ part of the construction for Theorem \ref{thm:main_l1_cs} can be written as $\Phi D$, where $\Phi_{h(i), i} = 1$ and all remaining entries are 0.
In the new construction for $\Pi_2$, $\Phi_{h(i), i}$ are i.i.d. samples from the Bernoulli distribution $\mathsf{Ber}(\varepsilon)$.
I.e., $\Phi_{h(i), i} = 1$ with probability $\varepsilon$ and $0$ otherwise.
All other parts of the construction are the same as in Theorem \ref{thm:main_l1_cs}.

We note that the proof for Theorem \ref{thm:main_l1_cs} can still go through for the new construction. 
The only difference occurs when proving Lemma \ref{lem:light_contraction_l1}.
In fact, the $\Pi_2$ part of the new construction for Theorem \ref{thm:main_l1_sparse} can be viewed as the following equivalent two-step procedure. 
For any given vector $y = Ux$, we first zero out each coordinate of $y$ with probability $1 - \varepsilon$, which results in a new vector $\overline{y}$, and then apply the $\Pi_2$ part of the embedding in Theorem \ref{thm:main_l1_cs}
on the new vector $\overline{y}$.

Now we show that with probability at least $1 - \exp \left( -\Omega(d^2\varepsilon) \right)$, we have
$$
\|\overline{y}_{d^2 + 1 : n}\|_1 \ge \Omega(\epsilon) \|y_{d^2+1:n}\|_1.  
$$
Notice that $\E[|\overline{y}_i|] = \varepsilon |y_i|$ and $\E[\overline{y}_i^2] = \varepsilon y_i^2$, which implies
$$
\E[\|\overline{y}_{d^2 + 1 : n}\|_1] = \varepsilon \cdot \|y_{d^2+1:n}\|_1
$$
and
$$
\sum_{i>d^2} \E[\overline{y}_i^2] = \sum_{i>d^2} \varepsilon y_i^2\le \varepsilon \| y_{d^2 + 1 : n}\|_1 \cdot \|y_{d^2 + 1 : n}\|_{\infty} = \varepsilon d^{-2}  \| y_{d^2 + 1 : n}\|_1^2.
$$
Thus, by Maurer's inequality in Lemma \ref{lem:maurer}, with probability at least $1 - \exp \left( -\Omega(d^2\varepsilon) \right)$, we have
$$
\|\overline{y}_{d^2 + 1 : n}\|_1 \ge \Omega(\epsilon) \|y_{d^2+1:n}\|_1.  
$$
The rest of the proof is identical to the original proof for Lemma \ref{lem:light_contraction_l1}.
Similarly, the same argument can also be applied to Theorem \ref{thm:main_lp_cs}. 
\begin{theorem}\label{thm:main_lp_sparse}
Suppose $1 \le p < 2$.
For any given $A \in \mathbb{R}^{n \times d}$, let $U$ be a $(d^{1 / p}, 1, p)$-well-conditioned basis for $A$. 
For any constant $0 < \varepsilon < 1$, there exists an $\ell_p$ oblivious subspace embedding over $O(d^2) \times n$ matrices $\Pi$ 
where each column of $\Pi$ has at most two non-zero entries and $1 + \varepsilon$ non-zero entries in expectation,
such that
with probability $0.99$,  for any $x \in \mathbb{R}^d$,
$$
\Omega(1) \|Ux\|_p \le \|\Pi Ux\|_p \le O\left((d \log d)^{1 / p}\right) \|Ux\|_p.
$$
\end{theorem}

The number of rows in Theorem \ref{thm:main_l1_sparse} and Theorem \ref{thm:main_lp_sparse} cannot be further reduced.
It is shown in \cite{nelson2013sparsity} (Theorem 16) that for any distribution over $r \times n$ matrices $\Pi$ such that any matrix in its support has at most one non-zero entry per column, if $\rank(\Pi A) = \rank(A)$ holds with constant probability, then $r = \Omega(d^2)$.
Now we sketch how to generalize this lower bound to distributions over $r \times n$ matrices for which each column has at most $1 + \varepsilon$ non-zero entries in expectation, for any constant $0 < \varepsilon < 1$.
Notice that such a lower bound already implies the number of rows of Theorem \ref{thm:main_l1_sparse} and Theorem \ref{thm:main_lp_sparse} are optimal up to constant factors, since any oblivious subspace embedding preserves the rank with constant probability. 

For each column in the matrix $\Pi$, by Markov's inequality, with probability at least $1 - \frac{1 + \varepsilon}{2}$, there will be at most one non-zero entry in that column.
By the Chernoff bound in Lemma \ref{lem:chernoff_bound}, with probability $1 - \exp(-\Omega(n))$, the number of columns in $\Pi$ with at most one non-zero entry is $\Omega(n)$.
Furthermore, the balls and bins analysis in the proof of Theorem 16 in \cite{nelson2013sparsity} can be applied to distributions over $r \times n$ matrices
such that for any matrix in the support of the distribution, 
the number of columns with at most one non-zero entry is $\Omega(n)$. Indeed, with constant probability the rank will drop if the
embedding matrix has $o(d^2)$ rows. 
This establishes the desired lower bound of $r = \Omega(d^2)$.

\subsection{Improving Sparsity Based on Truncation}\label{sec:truncating}
In this section, we show how to use a truncation argument to improve the construction in \cite{meng2013low}.

Before formally stating the construction, we first define the truncation operation. 
For a given parameter $\alpha > 0$, for any $x \in \mathbb{R}$, define
$$
\trunc{\alpha}{x} = \begin{cases}
\alpha & \text{if $x \in [0, \alpha]$} \\
-\alpha & \text{if $x \in [-\alpha, 0)$} \\
x & \text{otherwise} 
\end{cases}.
$$
Here we note some elementary properties of the truncation operation.
\begin{lemma} \label{lem:property_truncation}
For a given parameter $\alpha > 0$, for any $x \in \mathbb{R}$, we have
\begin{itemize}
\item $|\trunc{\alpha}{x}| \ge \alpha$.
\item $\trunc{\alpha}{x} - \alpha \le x \le \trunc{\alpha}{x} + \alpha$.
\end{itemize}
\end{lemma}
When applying the truncation operation to standard Cauchy random variables, the following properties are direct implications of Lemma \ref{lem:property_truncation} and the $1$-stability of standard Cauchy random variables. 
\begin{corollary}\label{lem:property_trunc_cauchy}
For $i \in [n]$, let $\{X_i\}$ be $n$ independent standard Cauchy random variables. The following holds.
\begin{itemize}
\item $\left| \trunc{\alpha}{X_i}\right| \ge \alpha$.
\item For any $a = (a_1, a_2, \ldots, a_n) \in \mathbb{R}^n$, 
$$
\|a\|_1 \cdot \hat{X} - \|a\|_1 \cdot \alpha \preceq
\sum_{i=1}^n a_i \cdot \trunc{\alpha}{X_i} \preceq
\|a\|_1 \cdot \hat{X} + \|a\|_1 \cdot \alpha,
$$
where $\hat{X}$ is a standard Cauchy random variable.
\end{itemize}
\end{corollary}
Now we are ready to state the main result of this section.
\begin{theorem}\label{thm:main_l1_trunc}
There exists an $\ell_1$ oblivious subspace embedding over $\widetilde{O}(d^4) \times n$ matrices $\Pi$ where each column of $\Pi$ has a single non-zero entry. The distortion $\kappa = \widetilde{O}(d)$.
\end{theorem}
Our embedding for Theorem \ref{thm:main_l1_trunc} is almost identical to the embedding for Theorem 2 in \cite{meng2013low} and the $\Pi_2$ part of the embedding for Theorem \ref{thm:main_l1_cs} and Theorem \ref{thm:main_l1_osnap}, except for replacing standard Cauchy random variables with truncated standard Cauchy random variables.
Let $R = \widetilde{O}(d^4)$ be the number of rows of $\Pi$. 
Here $\Pi$ can be written as $\Phi D : \mathbb{R}^n \to \mathbb{R}^{R}$, defined as follows:
\begin{itemize}
\item $h:[n] \to [R]$ is a random map so that for each $i \in [n]$ and $t \in [R]$, $h(i)=t$ with probability $1/R$.
\item $\Phi$ is an $R \times n$ binary matrix with $\Phi_{h(i),i} = 1$ and all remaining entries $0$.
\item $D$ is an $n \times n$ random diagonal matrix where $D_{i,i} = \trunc{\alpha}{X_i}$.
Here $\{X_i\}$ are i.i.d. samples from the standard Cauchy distribution and $\alpha < 1 / 4$ is a positive constant.
\end{itemize}

Now we sketch how to modify the proof of Theorem 2 in \cite{meng2013low} to prove the the distortion bound of our new embedding.

In the proof of Theorem 2 in \cite{meng2013low}, the authors define five events: $\mathcal{E}_U$, $\mathcal{E}_L$, $\mathcal{E}_H$, $\mathcal{E}_C$ and $\mathcal{E}_{\hat{L}}$.
Notice that for our new embedding, the event $\mathcal{E}_C$ is no longer needed, since by Corollary \ref{lem:property_trunc_cauchy}, the absolute values of standard Cauchy random variables are never smaller than $\alpha$ after truncation, where $\alpha$ is a small constant. 
We also change the number of rows of $\Pi$ to $O(d^4 \log^5 d)$, and the definition of the event $\mathcal{E}_{\hat{L}}$ is changed to $\| \Pi U^{\hat{L}} \|_1 \le O(1 / d \log^2 d)$ correspondingly. 

Lemma 16 and Lemma 22 in the proof for Theorem 2 in \cite{meng2013low} show that $\mathcal{E}_U$ and $\mathcal{E}_{\hat{L}}$ hold with constant probability.
The proofs for these two lemmas almost remain unchanged, except for replacing the $1$-stability of standard Cauchy random variables with the upper bound part of the ``approximate $1$-stability'' of truncated standard Cauchy random variables in Corollary \ref{lem:property_trunc_cauchy}. 

Lemma 13 is changed to the following: Given $\mathcal{E}_L$, for any fixed $y \in Y^L$, we have
$$
\Pr\left[\|\Pi y\|_1 \le \left(\frac{1}{4} - \alpha \right) \|y\|_1\right] \le \exp(- \Omega(d \log d)).
$$
The proof of the new version of Lemma 13 is also similar to the original proof, except for replacing the $1$-stability of standard Cauchy random variables with the lower bound part of the ``approximate $1$-stability'' of truncated standard Cauchy random variables in Corollary \ref{lem:property_trunc_cauchy}. 
This also explains why we need $\alpha$ to be a constant smaller than $1 / 4$.
Similarly, the constant $1 / 8$ in Lemma 14 also needs to be modified to reflect the changes in Lemma 13.

Finally, since the absolute values of standard Cauchy random variables are never smaller than $\alpha$ after truncation, Lemma 15 is changed to the following:
Given $\mathcal{E}_H$ and $\mathcal{E}_{\hat{L}}$, for any $y \in Y^H$ we have $\|\Pi y\|_1 \ge \Omega(\alpha) \|y\|_1$.
This finishes our modification to the proof of Theorem 2 in \cite{meng2013low}.

By applying the truncation argument to $p$-stable random variables, a similar result can be obtained for $\ell_p$ oblivious subspace embeddings. 
\begin{theorem}\label{thm:main_lp_trunc}
For $1 \le p < 2$, there exists an $\ell_p$ oblivious subspace embedding over $\widetilde{O}(d^4) \times n$ matrices $\Pi$ 
where each column of $\Pi$ has a single non-zero entry. The distortion $\kappa = \widetilde{O}(d^{1 / p})$.
\end{theorem}

%% file: appendix.tex
\appendix
\section{Missing Proofs in Section \ref{sec:pre}}\label{sec:proof}
\subsection{Proof of Lemma \ref{lem:pnorm_weighted_gaussian}}
\begin{proof}
$$
\E\left[\sum_{i=1}^n |a_iX_i|^p\right] = \sum_{i=1}^n |a_i|^p \E\left[|X_i|^p\right] = A_p \sum_{i=1}^n |a_i|^p,
$$
where $A_p = \E\left[|X_i|^p\right]$ is a constant which depends only on $p$.

Thus, by Markov's inequality, with probability at least $0.995$
$$
\left(\sum_{i=1}^n |a_iX_i|^p\right)^{1 / p} \le (200A_p)^{1 / p} \|a\|_p.
$$

There exists a constant $B_p$, which depends only $p$, such that
$$
\Pr[|X_i|^p < B_p] \le \frac{1}{400}.
$$
We let $Y_i$ be an indicator variable such that
$$
Y_i = \begin{cases}
1 & \text{if }|X_i|^p < B_p\\
0 & \text{otherwise}
\end{cases}.
$$
We know that $\E[Y_i] \le \frac{1}{400}$, which also implies $\E\left[\sum_{i=1}^n |a_i|^p \cdot Y_i|\right] \le \frac{1}{400} \|a\|_p^p$.
Thus by Markov's inequality, with probability at least $0.995$ we have
$$
\sum_{i=1}^n |a_i|^p \cdot Y_i  \le \frac{1}{2} \|a\|_p^p.
$$
Notice that
$$
\sum_{i=1}^n |a_iX_i|^p \ge B_p \sum_{i=1}^n |a_i|^p (1 - Y_i)
$$
Thus, with probability at least $0.995$,
$$
\left(\sum_{i=1}^n |a_iX_i|^p\right)^{1 / p} \ge \left( \frac{B_p}{2} \right) ^{1 / p} \|a\|_p.
$$
Thus, the lemma holds by taking $C_p = \max\left\{(200A_p)^{1 / p}, \left( \frac{2}{B_p} \right) ^{1 / p}\right\}$ and using a union bound.
\end{proof}
\subsection{Proof of Lemma \ref{lem:cauchy_upper}}
\begin{proof}
Let $\mathcal{E}_i$ be the event that $|X_i| \le \frac{n \log n}{\log \log n}$.
According to the cumulative density function of the standard Cauchy distribution, we have
$$
\Pr[\mathcal{E}_i] = 1 - \frac{2}{\pi} \arctan\left(n \log n / \log \log n \right) \ge 1 - \frac{2 \log \log n}{\pi n \log n}.
$$
Let $\mathcal{E} = \bigcap_{1 \le i}^n \mathcal{E}_i$.
By a union bound, $\mathcal{E}$ holds with probability at least $1 - \frac{2 \log \log n}{\pi \log n}$.
Next we calculate $\E[|X_i| \mid \mathcal{E}]$. 
Since the $X_i$ are independent, by using the probability density function of the standard Cauchy distribution,
$$
\E[|X_i| \mid \mathcal{E}] = \E[|X_i| \mid \mathcal{E}_i] = \frac{1}{\Pr[\mathcal{E}_i]} \frac{1}{\pi} \log\left(1 + \left( n \log n / \log \log n \right)^2\right) = O(\log n).
$$
Notice that conditioned on $\mathcal{E}$, $|X_i|$ are still independent. 
Furthermore, conditioned on $\mathcal{E}$, for any $i \in [n]$, $|X_i| \in [0, n \log n / \log \log n] $. 
Thus for sufficiently large $\constantupper_1$, by applying the Chernoff bound in Lemma \ref{lem:chernoff_bound} on $|X_i| \log \log n (n \log n)^{-1}$,
$$
\Pr\left[\sum_{i=1}^n |X_i| > \constantupper_1 n \log n \mid \mathcal{E} \right] \le 2^{-\frac{\constantupper_1 n \log n}{n \log n / \log \log n}} = 2^{-\constantupper_1 \log \log n}.
$$
Thus for sufficiently large $\constantupper_1$, 
\begin{align*}
&\Pr\left[\sum_{i=1}^n |X_i| \le \constantupper_1 n \log n \right] \ge \Pr\left[\sum_{i=1}^n |X_i| \le \constantupper_1 n \log n \mid \mathcal{E} \right] \cdot \Pr[\mathcal{E}] \\ 
\ge &\left(1 - 2^{-\constantupper_1 \log \log n}\right) \cdot \left(1 - \frac{2 \log \log n}{\pi \log n}\right) \ge 1 - \frac{\log \log n}{\log n }.
\end{align*}
\end{proof}
\subsection{Proof of Lemma \ref{lem:p_stable_lower}}
\begin{proof}
According to Lemma \ref{lem:nolan_tail}, there exists a constant $t_p \ge 1$ which depends only on $p$, such that for any $t \ge t_p$,
$$
\Pr[X_i > t] \ge \frac{c_p}{2}  t^{-p}.
$$
Thus for $t \ge t_p^p$
$$
\Pr[|X_i|^p > t] = \Pr[|X_i| > t^{1 / p}] = 2\Pr[X_i > t^{1 / p}] \ge c_p t^{-1}.
$$
For $i \ge 0$ and $j \in [n]$, we let $N_j^i$ denote the indicator variable such that
$$
N_j^i = \begin{cases}
1 & \text{if } |X_j|^p > 2^i t_p^p \\
0 & \text{otherwise}
\end{cases}
$$
and
$
N^i = \sum_{j=1}^n N_j^i.
$
We have that $\E[N_j^i] \ge 2^{-i}c_pt_p^{-p}$ and thus $\E[N^i] \ge n \cdot 2^{-i}c_pt_p^{-p}$.
According to the Chernoff bound in Lemma \ref{lem:chernoff_bound} we have $\Pr\left[N^i \ge n 2^{-i-1}c_pt_p^{-p}\right] \ge 1 - \exp\left(-n \cdot 2^{-i-3}c_pt_p^{-p}\right)$.
Let $l_{max}$ be the largest $i$ such that
\begin{equation}\label{equ:lmax}
\exp\left(-n2^{-i-3 }c_p t_p^{-p}\right) \le \frac{1}{2T}.
\end{equation}
By a union bound, with probability at least
$$
1 - \sum_{i=0}^{l_{max}} \exp\left(-n2^{-i-3 c_p t_p^{-p}}\right) \ge 1 - 1 / T,
$$
simultaneously for all $0 \le i \le l_{max}$, $N^i \ge n 2^{-i-1}c_pt_p^{-p}$, which implies
$$
\sum_{i=1}^n |X_i|^p \ge \sum_{i=0}^{l_{max}} 2^i t_p^p \cdot N^i / 2 \ge c_p / 4 \cdot l_{max} \cdot n .
$$
Solving (\ref{equ:lmax}) and substituting the value of $l_{max}$, for sufficiently large $T$ and $n$, with probability at least $1 - 1 / T$, 
$$
\sum_{i=1}^n |X_i|^p  \ge L_p n \log\left(\frac{n}{\log T}\right),
$$
where $L_p$ is a constant which depends only on $p$.
\end{proof}
\subsection{Proof of Lemma \ref{lem:cs_osnap_lp}}
\begin{proof}
Suppose $\Pi$ has $R$ rows and $s$ non-zero entries per column.
For the \textsf{CountSketch} embedding we have $R = O(d^2)$ and $s = 1$, 
while for the \textsf{OSNAP} embedding we have $R = O(B \cdot d \log d)$ and $s = O(\log_B d)$.
In either case we have $R \le O(d^2)$.

For $i \in [R]$, define  $B_i = \{j \mid j \le d^2 \text{ and }\Pi_{i, j} \neq 0\}$.
According to the Chernoff bound in Lemma \ref{lem:chernoff_bound}, with probability at least $1 - \exp(-\Omega(\omega d \log d))$, $|B_i| \le \omega d \log d$.
It follows by a union bound that with probability at least $1 - \exp(-\Omega(\omega d \log d)) \cdot R = 1 - \exp(-\Omega(\omega d \log d))$, 
simultaneously for all $i \in [R]$ we have
$|B_i| \le \omega d \log d$.
We condition on this event in the rest of the proof.

Notice that 
\begin{align*}
&\left|(\Pi \left(y_{1:d^2} \right))_i\right|^p 
\le \left(s^{-1 / 2}\sum_{j \in B_i} |y_j|\right)^p \\
\le & \left(s^{-1 / 2}(\omega d \log d)^{1 - 1 / p}\left(\sum_{j \in B_i} |y_j|^p\right)^{1 / p}\right)^p
= s^{-p/2}(\omega d \log d)^{p - 1} \left(\sum_{j \in B_i} |y_j|^p\right).
\end{align*}
Here the second inequality follows from Lemma \ref{lem:inter_norm} and $|B_i| \le \omega d \log d$.
For each $j \in [d^2]$, the number of $i \in [R]$ for which $j \in B_i$ is exactly $s$, which implies
$$
\sum_{i=1}^{R} \left|(\Pi \left(y_{1:d^2} \right))_i\right|^p 
\le s^{-p/2}(\omega d \log d)^{p - 1} \sum_{i=1}^R \sum_{j \in B_i} |y_j|^p
= s^{1 - p / 2} (\omega d \log d)^{p - 1}\sum_{j=1}^{d^2} |y_j|^p.
$$
Thus, 
$$
\|\Pi \left(y_{1:d^2} \right)\|_p \le s^{1 / p - 1 / 2} (\omega d \log d)^{1 - 1 / p} \|y_{1:d^2}\|_p \le s^{1 / p - 1 / 2} (\omega d \log d)^{1 - 1 / p} \|y\|_p.
$$
\end{proof}